\documentclass[11pt]{article}
\usepackage[T1]{fontenc}
\usepackage[utf8]{inputenc}
\usepackage{lmodern}
\usepackage[margin=1in]{geometry}
\usepackage{amsmath,amssymb,amsthm,mathtools}
\usepackage{graphicx}
\usepackage{booktabs}
\usepackage{longtable}
\usepackage{caption}
\usepackage[hidelinks]{hyperref}
\hypersetup{pdftitle={The Institutional Window: How Contract Law Bounds Liability Signaling of Human Fallback Capability under Generative AI},pdfauthor={Andreas Bauer},pdfsubject={Review of Law & Economics submission, v11.6},pdfkeywords={liability signaling; limitation of liability; standard-terms control; penalty clauses and agreed damages; comparative law and economics; artificial intelligence; credence goods}}
\usepackage[round,authoryear]{natbib}
\usepackage{setspace}
\usepackage{enumitem}
\usepackage[protrusion=true,expansion=false]{microtype}

\newtheorem{proposition}{Proposition}
\newtheorem{corollary}{Corollary}
\newtheorem{lemma}{Lemma}
\theoremstyle{definition}
\newtheorem{definition}{Definition}
\newtheorem{remark}{Remark}

\newcommand{\Lam}{\Lambda}
\newcommand{\sfl}{\underline{s}}
\newcommand{\sbar}{\bar{s}}

\newcommand{\E}{\mathbb{E}}
\newcommand{\1}{\mathbf{1}}

\providecommand{\anonymousbuild}{0}
\newif\ifanon
\ifnum\anonymousbuild=1 \anontrue \else \anonfalse \fi

\title{\textbf{The Institutional Window:\\ How Contract Law Bounds Liability Signaling of\\ Human Fallback Capability under Generative AI}}

\ifanon
  \author{}
\else
  \author{Andreas Bauer\\[2pt]{\normalsize Aegis Compliance and Strategies O\"U, Tallinn, Estonia}\thanks{ This paper builds on two companion working papers by the author, \citet{bauer2026a} and \citet{bauer2026b}; Section~\ref{sec:provenance} states the provenance of every model component. Simulation and calibration code, together with an interactive simulator, are openly available in the replication package at \url{https://doi.org/10.6084/m9.figshare.33174593}.}}
\fi

\date{September 2026}

\begin{document}
\maketitle

\begin{abstract}
\noindent Generative AI makes expert output less informative about the human
fallback capability retained by its provider.  This paper studies when a
liability commitment can signal that endogenous, perishable capability.  Four
legal primitives map a posted cap and agreed-damages term into retained
exposure: litigation viability, the penalty doctrine, displacement to the state
or an indemnity pool, and mandatory limits on contractual caps.  The exact
image of that map is decisive.  When the positive image is connected and
commitments below the litigation threshold remain legally admissible, retained
exposure has the message set $\{0\}\cup[F,C]$.  Low types then pool at zero; intermediate types separate on
an endogenously anchored Riley schedule; and high types may pool at the ceiling.
We derive both cutoffs, specify off-path beliefs, and verify all profitable
deviations, including boundary deviations.  If law instead removes the
zero-exposure region, separation begins at the bottom type.  A statutory floor
can therefore remove lower pooling while also moving the positive schedule
toward the ceiling, so its total effect is not globally monotone.  Calibration
to five occupations and twelve institutional configurations illustrates the
regimes but does not estimate causal legal effects.  The results reject
a general ranking by legal family or contracting channel; they support only
configuration-specific comparisons.  Separately, the engagement share needed
to maintain fallback skill rises as AI failures become rarer.  Liability law
can thus affect whether capability is observable, while the production of that
capability remains an operational investment problem.

\vspace{0.6em}

\vspace{0.3em}
\end{abstract}

\vspace{0.4em}
{\small
\noindent\textbf{Data availability.} \ifanon No proprietary or personal data were used. All reported numbers, tables and figures are generated by a deposited replication package (model and calibration modules, every analysis script with a script-to-table map, the claim-robustness engine with a verification harness, and an interactive browser-based simulator); the archive DOI is withheld for double-anonymous review and will be restored upon acceptance. Numeraire: $v=1\equiv$ EUR 5{,}000.\else No proprietary or personal data were used. All reported numbers, tables and figures are generated by the replication package \citep{deposit2026}, deposited at \url{https://doi.org/10.6084/m9.figshare.33174593}, which includes the model and calibration modules, every analysis script with a script-to-table map, the claim-robustness engine together with a verification harness that regenerates every reported quantity, and an interactive browser-based simulator. Numeraire: $v=1\equiv$ EUR 5{,}000.\fi

\noindent\textbf{Declarations of interest.} \ifanon Withheld for double-anonymous review; the full declaration accompanies the title page and will be restored upon acceptance.\else The author is Managing Director and Founder of Aegis Compliance and Strategies O\"U, a compliance and strategy advisory firm, and is the author and publisher of \emph{Diebstahlsicher} (Bauer Advanced Network Solutions KG, Vienna, September 2026, ISBN 978-3-9506352-0-1), a German-language practitioner book that draws on this line of work; both are disclosed in the interest of transparency.\fi

\noindent\textbf{Funding.} No external funding was received for this work.

\noindent\textbf{Ethics approval and consent.} Not applicable: the study involves no human participants, no animals and no personal data.

\noindent\textbf{Code availability.} The replication package regenerates every reported quantity; see Data availability.
}

\newpage
\section{Introduction}

A structural engineer signs off on a load calculation the AI produced. A radiologist reads a study the algorithm has already flagged. A senior developer approves a pull request written by a coding agent. In each case the artifact is excellent, and in each case the client---the building owner, the patient, the customer---cannot tell from the artifact whether the professional retains the capability to catch the algorithm when it is wrong.

This is the operational problem generative AI creates for expert services. It is not primarily a productivity problem; the productivity evidence is by now substantial and mostly positive \citep{brynjolfsson2025,noy2023,dellacqua2025,cui2026}. It is a problem of \emph{capability maintenance under invisibility}. The capability that matters---what \citet{singh2026} call the human fallback---is built by having people work through cases rather than approve them, it depreciates when they do not, and it is invisible in the output precisely because the AI has made the output uniformly good. \citet{budzyn2025} provide the sharpest evidence that the depreciation is real: in a multicentre observational study, adenoma detection in non-AI colonoscopy fell from 28.4\% to 22.4\% after endoscopists had been exposed to AI assistance, among physicians with more than two thousand procedures each.\footnote{The causal reading of that study has been contested in subsequent correspondence and the paper carries a corrigendum; we treat it as one calibrating datum among several, not as identification. See Section~\ref{sec:calib} and Appendix~C.}

Two companion papers establish the market-level consequence. \citet{bauer2026a} shows that as AI compresses the range over which buyers can discriminate artifacts, production-cost-based signals lose their informational content, while outcome-contingent commitments---warranties, indemnities, penalty clauses---retain theirs, because their expected cost is invariant to the cost of producing the artifact. The premise has direct evidence. \citet{galdinsilbert2025} study an online
labour market around the introduction of large language models: the coefficient
on the written application as a predictor of being considered falls from
$0.0967$ to $0.0340$, and its power to predict successful completion falls from
$0.0158$ to essentially zero. Observable characteristics, platform reputation
included, explain around three per cent of the variance in the ability the
application used to convey. In their counterfactual without that signal, workers
in the top ability quintile are hired $19\%$ less often, those in the bottom
quintile $14\%$ more often, and average wages fall by $5\%$. That is the
production-side collapse this paper takes as given, measured. \citet{bauer2026b} makes the certified type endogenous: because a provider with higher fallback skill fails less often when the AI fails, the expected cost of a liability pledge is strictly decreasing in that skill, which restores the Spence--Mirrlees condition on a variable the firm \emph{chooses}. The least-cost separating schedule is
\begin{equation}\label{eq:schedule}
\Lam^*(s) \;=\; c_0 \;+\; \frac{\chi}{\theta\pi}\,\ln\!\left[\frac{1-\rho(s_F)}{1-\rho(s)}\right],
\end{equation}
where $s$ is fallback skill, $\rho(\cdot)$ the rescue technology, $\pi$ the AI failure rate, $\theta$ ex post verifiability, $\chi$ the client's stake, and $c_0$ the fixed cost of enforcement.

Both papers stop at the same place. They establish that the mechanism \emph{exists}. They say nothing about where it \emph{works}. Their parameters are, in the authors' own words, identified in sign and order of magnitude but not estimated; their institutional content is a single fixed-cost term $c_0$ and a single solvency cap $\bar L$. Yet the entire mechanism runs through legal institutions that vary enormously---across occupations, because verifiability and stakes vary, and across jurisdictions, because the law of remedies varies. A liability pledge that separates types in a German software contract may be legally void in a British one, economically unenforceable in an American one, and simply irrelevant in German child protection, where the state assumes liability with discharging effect and the payer is not the beneficiary.

\subsection{What this paper does}

The question is legal before it is economic. A liability pledge that separates types in a German software contract may be void in the standard-terms channel of the same jurisdiction, unenforceable as a penalty in an English one, worthless in an American one, and simply irrelevant in German child protection, where the state assumes liability with discharging effect and the payer is not the beneficiary. The economics of the signal are settled; what is not settled is which legal configurations admit it. We therefore hold the signaling mechanism fixed and vary the law.

We take the mechanism as given and ask where it survives. The contribution has four parts.

\textbf{First, an institutional wedge.} We introduce a mapping $\Lam(K;J,v)$ from the cap a firm \emph{posts} to the retained exposure that \emph{carries information}, indexed by a jurisdiction vector $J$. Four legal primitives enter: the cost-allocation rule (which determines when a claim is worth bringing), the penalty doctrine (which determines how large a pledge may lawfully be), the displacement of private liability by state liability or pooled indemnity (which determines whether the provider bears the pledge at all), and mandatory limits on contractual liability (which determine how small a cap may lawfully be). Single crossing survives the wedge if and only if $\Lam$ is strictly increasing over a non-degenerate range (Proposition~\ref{prop:sc}).

\textbf{Second, a two-sided institutional compression.}  The wedge confines
separation to $[a,c^*]$ (Propositions~\ref{prop:tax}--\ref{prop:bunch}).  The
lower cutoff is not obtained by intersecting an unconstrained schedule with a
floor.  When zero exposure is admissible, cutoff indifference determines $a$
and low types pool at zero.  The upper cutoff incorporates the belief premium
offered by the ceiling pool.  This three-region construction is the paper's
central correction and is verified against all pairwise and off-path
deviations.

\textbf{Third, occupation $\times$ jurisdiction calibration.} We calibrate five occupations (software development, legal research, radiology, audit, social work) against published evidence on AI error rates, deskilling and enforcement costs, and five legal systems in twelve institutional configurations: Germany and Austria in their negotiated and standard-terms channels---which mandatory law treats differently, Section~\ref{sec:model}---Austria additionally split by contract vintage across the 2007 reform, Switzerland as a band, the United Kingdom in its private and NHS channels, and the United States in B2B and B2C configurations. The resulting diagnosis map (Figure~\ref{fig:heatmap}) classifies each cell and identifies which institutional margin binds.

\textbf{Fourth, a dynamic boundary condition.} The capability the signal certifies is produced by engagement and eroded by its absence, and Section~\ref{sec:dynamics} derives the engagement share that holds it constant. That share rises as the machine improves, which is the operational counterpart of the certification problem. We do not extend the analysis to a simulated market: whether an institutional regime determines the long-run capability stock turns on whether a public performance record substitutes for certification, and answering that requires a model of the record which this paper does not have. We state the boundary and leave it open.

\subsection{Three results worth stating up front}

\textbf{The penalty doctrine converts contracting into operations.} Separation requires exposure that compensates for imperfect verifiability: the client's expected recovery must satisfy $\theta G\geq v$, a gross-up of the loss by the inverse of the probability that fault can be proven \citep[Prop.~3, the $\varphi D\geq v$ condition]{bauer2026a}. Common-law jurisdictions restrict this channel. Under the penalty doctrine, agreed damages out of all proportion to any legitimate interest in enforcement are struck down; a multiple set purely to compensate for evidentiary risk is vulnerable as a penalty. Formally, separation requires $m \geq 1/\theta$, where $m$ is the enforceable multiple of loss. With $m=1$ in the United States and a central $m=3$ in the post-\emph{Cavendish} United Kingdom---the \emph{Houssein} litigation upheld a default rate at four times the standard rate as non-penal, [2025] EWHC 2749 (Ch), affirmed [2026] EWCA Civ 830; no point value in the defensible range $1.2$--$4$ is compelled---our calibration finds the sub-condition $m\geq1/(\theta\kappa)$ satisfied in four of ten common-law occupation--jurisdiction pairs---the four market occupations in the United Kingdom, and none in the American B2B channel. The full condition of Proposition~\ref{prop:penalty}, $\theta\bar G\geq v$ with $\bar G=\min\{m\kappa v,\bar L_{\mathrm{solv}},\bar L_{\mathrm{law}}\}$, holds in two of the ten: British software development and British audit. In British legal research and radiology the doctrinal multiple is ample and the solvency ceiling binds first ($\theta\bar G=2.00$ against $v=2.50$, and $2.80$ against $3.00$), which is the case Proposition~\ref{prop:penalty} warns about and the reason the sub-condition is only sufficient where the other ceilings are slack. The condition acquires bite in the civil law as well. With the Austrian and Swiss multiples finite (Section~\ref{sec:model}), the gross-up is contractually unavailable in an Austrian negotiated mandate for legal research at the central post-reform multiple ($1/\theta=2.00>1.5$), and in the German standard-terms channel the case-group ceiling lies below every calibrated $1/\theta$ at the central encoding $m\approx1.05$, so at the centre it fails in every market occupation. Like the Austrian and Swiss statements this is a band statement: over the declared range $[1.00,1.30]$ the sub-condition acquires bite in audit above $m=1.176$ and in software development above $m=1.250$, and fails throughout only in legal research and radiology. At the lower endpoint of the English range the British count falls to one (audit), so the condition's bite in England is itself a band statement; in the American B2B channel it binds at every $\theta<1$. The dividing line runs between configurations, not between legal families. The implication is not that common-law providers cannot signal, but that where the condition fails they must signal through a different channel: by \emph{raising $\theta$}---investing in logging, audit trails, acceptance criteria and forensic reconstruction---rather than by raising $L$. A doctrine of contract law thus determines whether a firm's optimal response is legal or operational. The claim is confined to the liquidated-damages channel the model prices; insurance, service credits and audit rights are alternative carriers within the admissible instrument class of Proposition~\ref{prop:sc}.

\textbf{Verifiability investment has opposing boundary effects.} Raising
$\theta$ relaxes the upper constraint but, with a zero-exposure pool, raises
the lower cutoff.  The separating width can therefore rise or fall.  The paper
states the two derivatives separately and does not infer an unconditional
investment ranking.

\textbf{The reliability paradox.} The engagement share needed to hold fallback skill constant is $h_{\min}=\gamma/(\varphi\pi+\gamma)$: as the AI becomes more reliable ($\pi\downarrow$), \emph{more} human engagement is required, because failures are the only learning events. Radiology, with $\pi=0.10$, requires $h_{\min}=0.92$; software development, with $\pi=0.30$, requires $0.75$. Improving AI reliability does not relax the capability-maintenance constraint; it tightens it. This is the operational counterpart to the observation in \citet{bauer2026b} that better unassisted AI narrows the exclusion band from below---the effects compound.

\subsection{\ifanon Relation to prior work\else Relation to the companion papers\fi}

\ifanon We are explicit about what is new.\else We are explicit about what is new; the two papers below are companion work by the same author.\fi\ \citet{bauer2026a} establishes the collapse of production-side signals and the survival of the liability class; \citet{bauer2026b} makes the type endogenous, derives \eqref{eq:schedule}, the two-audience targeting result and the free-riding threshold. Neither result is re-derived here; both are used. Everything in Sections~\ref{sec:wedge}--\ref{sec:dynamics}---the wedge, the two-sided window, the penalty-doctrine condition, the third-payer threshold, the $\theta$-investment reversal, the reliability paradox, and the entire calibration---is new. Table~\ref{tab:provenance} (Section~\ref{sec:provenance}) states the provenance of every component.

\section{Related literature}\label{sec:lit}

\paragraph{Law and economics of limited liability and agreed damages.} The paper's primary home is the economic analysis of contractual remedies. Damage measures and the incentives they create are the subject of \citet{shavell1980} and \citet{cooter1985}; the cost-allocation rule that determines whether a claim is worth bringing is analysed in \citet{shavell1982}; \citet{shavell2004} and \citet{faure2009} survey the tort and contract channels. The enforceability of agreed damages has its own literature: \citet{rea1984} and \citet{edlinschwartz2003} ask when supra-compensatory stipulations are efficient, and \citet{degeestwuyts2000} set out the comparative position that civil-law systems enforce penalties subject to judicial moderation while common-law systems strike them. \citet{schaeferott2022} give the civil-law treatment of standard-terms control on which our floor rests. Our contribution to this literature is to read three of these institutions not as devices for allocating risk or deterring breach but as constraints on an \emph{informational} use of the same clause, and to show that they bind from opposite sides. That reading also puts a familiar protective instrument in an unfamiliar light: the mandatory minimum liability of standard-terms control belongs, informationally, with the minimum-quality standard of \citet{leland1979} and \citet{ronnen1991} and with the crowding-out of voluntary by mandatory disclosure in \citet{fishmanhagerty2003} and \citet{benshaharschneider2011}, rather than with the consumer-protection rationale that produced it. On why one-sided terms survive competition, see \citet{katz1990} and \citet{bebchukposner2006}. On liability for AI specifically, \citet{buiten2024} analyses the product-liability channel and \citet{guerra2022a,guerra2022b} the allocation of robot liability; both take the machine as the locus of fault, where we take the human capability behind it as the object to be certified.

\paragraph{Verification, enforcement and the design of private remedies.} The
closest recent work shares our two central objects---a contract term and a
verification technology---and combines them the other way round.
\citet{likas2026} studies trade-based money laundering as endogenous contract
design under verification frictions: the firm chooses documentary complexity so
as to outrun a verifier whose capacity does not scale with it, and the
comparative static of interest is that raising audit intensity or penalties can
increase opacity rather than deter. There the contract term conceals private
information from an adversarial enforcer, the firm has no hidden quality
dimension to sort on, and the message space is unbounded. Here the contract term
reveals private information to a Bayesian buyer, the type is a capability the
firm has itself built, and the message space is closed on both sides by positive
contract law. The verifiability parameter accordingly enters with opposite
effect: it constrains the sender in \citet{likas2026} and prices the signal
here, which is why raising $\theta$ widens the window rather than narrowing the
firm's room for manoeuvre. The two papers treat law as an ex-post detection
apparatus and as an ex-ante restriction on what can be said, and we read them as
complements.

The remaining points of contact are separable.
\citet{martinlapoirie2022} models medical treatment as a credence good in which
the patient detects error imperfectly, and shows how that detection probability
governs the deterrent force of alternative liability rules; the same object
enters our schedule multiplicatively and governs the width of the window rather
than the strength of deterrence. \citet{kimkim2024} revisits the invariance of
liability rules in oligopoly, where liability is a rule imposed on the firm
rather than a term the firm posts and is observed choosing.
\citet{xuejiang2024} derives efficient private penalty schemes between an
expert firm and a lay customer; unsuitability there is verifiable ex post, which
is the assumption $\theta<1$ removes, and the schemes are not tested against the
penalty doctrine that bounds them. \citet{karampatzos2023} treats mandatory
consumer rules as a device for reducing informational asymmetry; our result runs
in the opposite direction, since a mandatory floor truncates the lower end of
the signalling interval and can extinguish separation.
\citet{gabuthy2026} models litigation cost as a selection device among claims of
given size, where we use a fixed cost block to exclude whole segments from the
signalling market. \citet{avrahamporat2023} sets out how insurance displaces
loss rather than minimising it; $\psi$ is the formal counterpart, with the scope
restriction recorded in Section~\ref{sec:conclusion}. The signalling strand in
this literature has been experimental and has taken the type as drawn---
\citet{pecorino2019} on costly voluntary disclosure and \citet{solomon2026} on
learning in a pretrial signalling game---so the construction this paper needs,
a separating schedule on a type the sender manufactures, has no precedent there.

\paragraph{Human--AI division of labor in service operations.} Our starting point is the emerging operations literature on how firms should allocate work between algorithms and people. \citet{cohen2026} set out the community's agenda for co-adaptation and shared decision authority; \citet{daitayur2022} identify the buy-in conditions for AI-augmented delivery; \citet{daising2025} model the physician's decision to invoke AI under alternative liability regimes---the closest antecedent to our liability channel, though with an exogenous competence type. \citet{boyaci2024} show that machine input shifts the human error distribution rather than uniformly improving it, and \citet{devericourt2023} show that a supervisor may never learn whether the machine is better, which is precisely the verification friction our $\theta$ parameter measures. We differ from all of these in making the human capability a \emph{market-observable} object that must be certified to a third party.

\paragraph{Learning, forgetting and skill as a state variable.} That capability depreciates is old news in operations: \citet{argote1990} and \citet{benkard2000} establish organizational forgetting; \citet{gans2002} and \citet{arlotto2014} embed learning and turnover in staffing policy; \citet{kc2013} and \citet{staats2012} document learning from experience in service settings. In all of this work, erosion is a by-product of interruption or turnover. The novelty of the fallback problem---and the reason \citet{singh2026} is the direct antecedent of our dynamics---is that erosion becomes a \emph{choice variable}: the firm sets the engagement share $h$, and with it the drift of the skill state. We take that dynamic as given and ask what makes the firm willing to choose $h>0$.

\paragraph{Warranties, performance contracting and credence services.} Liability appears in operations mainly as an incentive or risk-allocation instrument: \citet{kim2007} model performance-based contracting as a multitask agency problem. In economics it appears as a screening device over an exogenous type \citep{spence1977}. Our object is neither, and its nearest neighbours require explicit demarcation. The warranty-signaling canon---\citet{spence1977}, \citet{galor1989}, \citet{lutz1989}, and the money-back guarantee of \citet{moorthy1995}---treats coverage as an \emph{addition} above a finite default; \citet{lutzpadmanabhan1995} ask why observed warranties are minimal, and \citet{spier1992} reads contract terms generally as signals. Our instrument is the dual: a \emph{cap subtracted from an unlimited default}, whose admissible range is set by mandatory law rather than by demand. \citet{dks2011} establish experimentally that liability, not verifiability, disciplines credence markets---the moral-hazard complement to the adverse-selection channel here. The pledge, in short, is a signal about an endogenously produced operational asset, and the question is whether the legal system lets that signal through. \citet{debo2008} and \citet{veeraraghavan2011} bring credence-good and quality-inference problems into operations; \citet{dulleck2006} survey the economics. Where that literature asks whether the expert over-treats, we ask whether the expert has retained the capacity to treat at all.

\paragraph{Professional service operations.} The occupations we calibrate are customer-intensive expert services in the sense of \citet{anand2011}, organized around the gatekeeper--referral structure formalized by \citet{shumsky2003} and extended by \citet{dada2025}. \citet{tan2014} document the workload--productivity relation that our engagement drag captures in reduced form. Our contribution to this stream is to add a market layer in which the firm's staffing and engagement policy must be \emph{communicated} to clients under conditions where the output no longer communicates it.

\paragraph{Institutions.} There is a substantial law-and-economics literature on cost-allocation rules and on sovereign immunity, and a health-economics literature on third-party payment, but to our knowledge none of it has been connected to competence signaling under AI. Our treatment of the penalty doctrine follows the standard contrast between civil-law enforceability of contractual penalties and the common-law rule against penalties as restated in \emph{Cavendish Square Holding BV v Talal El Makdessi} [2015] UKSC 67. On mandatory standards displacing voluntary separation, the antecedents are \citet{leland1979} and \citet{ronnen1991}; on mandatory versus voluntary disclosure, \citet{fishmanhagerty2003}.

\section{Model}\label{sec:model}

\subsection{Primitives}

A unit mass of providers serves clients in a credence-good market. Each provider employs workers whose \emph{fallback skill} $s\in[0,\alpha]$ measures the probability-weighted capacity to recover a case when the AI fails; $\alpha\in(0,1)$ is the AI capability frontier. Time is discrete. In each period, a client arrives with a problem of value $v$; the AI attempts it and fails with probability $\pi$. Conditional on failure, the human recovers the case with probability
\begin{equation}\label{eq:rho}
\rho(s)\;=\;\rho_{\max}\left(\frac{s}{\alpha}\right)^{a},\qquad a\in(0,1),
\end{equation}
so that $\rho'>0$ and $\rho''<0$. Concavity is what drives the targeting results of \citet{bauer2026b} and, as we show, the shape of the window.

The firm chooses an \emph{engagement share} $h\in[0,1]$: the fraction of cases in which a human works the problem rather than approving the machine's output. Skill evolves as in \citet{singh2026},
\begin{equation}\label{eq:skill}
s_{t+1}\;=\;\mathrm{Proj}_{[0,\alpha]}\Bigl\{s_t+(\alpha-s_t)(1-\pi)\bigl[\varphi\pi h_t-\gamma(1-h_t)\bigr]\Bigr\},
\end{equation}
with learning rate $\varphi$ and erosion rate $\gamma$. The projections are not cosmetic. Without the lower one the smooth part is a linear recursion in $\alpha-s_t$ with factor $1+\gamma(1-\pi)>1$ whenever $h_t=0$, so the state leaves $[0,\alpha]$ from below within a few periods---in the software calibration, by the third. An upper projection is needed as well, because the smooth increment can exceed $\alpha$ whenever $(1-\pi)\varphi\pi>1$: at $\alpha=1$, $s=0.5$, $\pi=0.5$, $\varphi=10$, $\gamma=1$, $h=1$ the raw step returns $1.75$. Every calibrated vector satisfies $(1-\pi)\varphi\pi\le1$---the largest value in the sample is $0.126$---so the upper projection is inert in this paper, but the state space is not invariant without it. Zero is absorbing only under $h\le h_{\min}$: at $s=0$ with $h=1$ the software calibration returns $s'=0.111>0$. Engagement is costly: it imposes an output drag $\delta h(\alpha-s)$.

Clients cannot observe $s$. They observe the posted liability cap $K\geq0$; what they price is the retained exposure per verified failure that the cap, the penalty doctrine and the enforcement regime jointly leave with the provider. Ex post verifiability---the probability that a failure can be attributed and proven---is $\theta$. The client's stake is $\chi=\omega\,\zeta\,v\,\pi$, where $\zeta$ is the provider's share of surplus and $\omega\in[0,1]$ is a \emph{payer-alignment weight} introduced below.

\subsection{The institutional wedge}\label{sec:wedge}

A jurisdiction is a vector $J=(\text{rule},c_0,c_1,a_c,m,\psi,\underline{L},\bar L_{\text{law}},\kappa)$. The central construct of this paper is the map from the posted cap to the retained exposure.

\paragraph{The instrument.} The posted object is not a voluntary add-on to a finite default; it is a \emph{limitation-of-liability cap} $K$ set against an unlimited fault-based default. Under German law the default is full statutory damages (\S\S~249~ff., 280 BGB); a contractual cap subtracts from it, and standard-terms control determines which subtractions are enforceable. Because a cap voided by \S~307 BGB is not reduced to the permissible level but falls away entirely---no \emph{geltungserhaltende Reduktion} (BGHZ 84, 109), the gap being filled by dispositive law (\S~306(2) BGB)---an attempted cap below the lawful minimum returns the provider to unlimited exposure, which is dominated for every type. The enforceable choice set is therefore $\{K\geq f v\}\cup\{\varnothing\}$: a truncation of the message space that is \emph{derived} from the void-snap, not assumed about buyer inference. Common law runs the same instrument from the other side. In the United States, commercial limitation-of-remedy clauses may be enforceable under heterogeneous state-law doctrines; UCC \S~2-719 is illustrative for goods transactions only. In England, negotiated B2B caps may be upheld where the UCTA reasonableness test is met (\emph{Watford Electronics v Sanderson} [2001] EWCA Civ 317). The common-law benchmark therefore has no mandatory lower cap floor; its binding restriction in the model concerns agreed-damages add-ons, through which exposure \emph{above} compensatory loss---the provability gross-up of Proposition~\ref{prop:penalty}---must travel. The civil-law cap and the common-law add-on are two margins of the same exposure variable, restricted from below and from above respectively. We accordingly work directly in the space of retained exposure per verified failure, $\Lam$, with $K$ the posted cap; Lemma~\ref{lem:impl} shows this is without loss. What the companion papers call the liability pledge is, in institutional dress, exactly this retained exposure; the contractual carrier is the cap. This construction also answers an identification concern one might raise against treating a statutory default and a voluntary pledge as one signal: under an unlimited default there is no voluntary increment to post---the only margin of choice is how little to limit, and the signal is the cap.

\begin{definition}[Institutional wedge]\label{def:wedge}
For a posted cap $K$, value at stake $v$ and jurisdiction $J$, the \emph{retained (signal-bearing) exposure} is
\begin{equation}\label{eq:wedge}
\Lam(K;J,v)\;=\;\underbrace{(1-\psi)}_{\text{displacement}}\cdot\underbrace{\min\{K,\;m\,\kappa v\}}_{\text{penalty doctrine}}\cdot\underbrace{\1\{\text{enforcement viable}\}}_{\text{cost rule}} .
\end{equation}
\end{definition}

The three factors correspond to three distinct legal mechanisms.

\emph{(i) Displacement} $\psi\in[0,1]$. Where liability is transferred to the state with discharging effect---German official liability under Art.~34 GG in conjunction with \S~839 BGB, the Austrian \emph{Amtshaftungsgesetz}, the Swiss \emph{Verantwortlichkeitsgesetz}---the individual provider cannot bind itself, and the fraction $\psi$ of any pledge is economically borne by someone else. Recourse against the official is confined to intent and gross negligence, so $\psi$ is close to but strictly below one in Germany and Austria, and equal to one in Switzerland, where direct action against the official is excluded outright. Crucially, $\psi$ also captures a mechanism that is legally unrelated but economically identical: \emph{pooled indemnity}. Under NHS Resolution's clinical negligence scheme, liability is state-pooled with no provider-specific premium differentiation, which in the language of \citet[Prop.~8, retained and individually priced risk]{bauer2026a} is a pooled-premium regime with a zero deductible. Immunity and pooling are observationally distinct legal institutions with identical signaling consequences---a prediction only the signaling perspective generates.

\emph{(ii) The penalty doctrine} $m$. The multiple is channel-indexed, because the two German control regimes attach to mutually exclusive contracting channels. Between merchants an individually agreed penalty escapes judicial reduction (\S~348 HGB, barring \S~343 BGB): in the negotiated channel $m=\infty$ is the stylized benchmark---not literal, since in extreme cases \S~242 BGB caps even the merchant penalty (BGH I ZR 168/05). We assign that boundary no numerical value, because the case law identifies none; the benchmark should be read as the absence of a doctrinal ceiling in the calibrated range, and every German negotiated result is conditional on the \S~242 threshold lying above the solvency ceiling. Standard-terms penalties face \S~307 review with case-group ceilings (construction: 5\% of the contract sum, BGHZ 153, 311; reaffirmed BGH VII ZR 42/22; employment: approximately one month's salary): in the standard-terms channel we encode total enforceable exposure---compensatory recovery plus the penalty ceiling---as $m=1.05$ with a sensitivity band $[1.00,1.30]$. Because destructive control voids rather than reduces an offending clause, a rational firm posts within the valid range, so the nullity boundary acts on the \emph{chosen} penalty exactly as a cap. A single posted term cannot be simultaneously negotiated (for the unbounded multiple) and pre-formulated (for the floor): \S~305(1) s.~3 BGB makes the channels mutually exclusive, and we therefore report German---and, below, Austrian and Swiss---columns per channel. The high-floor, unbounded-multiple configuration is an across-channel construct, not a single-contract one.

Austria carried the same merchant rule until the Handels\-rechts-\"An\-de\-rungs\-ge\-setz (BGBl I 2005/120) repealed \S\S~348 and 351 HGB with effect from 1 January 2007, restoring the moderation right of \S~1336(2) ABGB to penalties promised by entrepreneurs; \S~907(18) UGB attaches the applicable regime to the time of contract formation, not of judgment, so the Austrian multiple is a step function of contract vintage: $m_{\text{AT}}=\infty$ for pre-2007 contracts and finite thereafter. The prevailing view holds the restored moderation right mandatory in favour of entrepreneurs; a dissenting view treats it as waivable in advance, and the Supreme Court has not resolved the point after the reform. Austrian moderation is anchored in the actual or ex-ante possible damage rather than in a multiple of the contract value, so no point value is a measurement; we centre the post-reform multiple at $1.5$ with a band of $1.2$--$2.0$ and report full band sensitivity. Cohort assignment and the behaviour of courts across the discontinuity are the subject of companion work \citep{bauer2026g}.

Switzerland has never known a merchant exception, and art.~163(3) OR is mandatory (BGE 143 III 1 E.~4.1, referring to BGE 133 III 201 E.~5.2). Its doctrinal multiple is therefore finite; its \emph{effective} level is a separate matter, since reduction proceeds only on the debtor's plea, the debtor bears the burden of pleading and proof (BGE 143 III 1; BGE 114 II 264 E.~1b; BGE 133 III 43 E.~4.1), and courts intervene only at gross disproportion. We compute throughout on the doctrinal multiple $m^{\text{doct}}$ and report the Swiss column as a band over $m_{\text{CH}}\in[1.5,3.0]$---no point value exists anywhere in our record---stating for every cell the exact threshold above which it is invariant (Corollary~\ref{cor:mx}). The divergence between $m^{\text{doct}}$ and the effective multiple $m^{\text{eff}}$ is measured in the companion work.

Common law strikes agreed damages out of all proportion to any legitimate interest in enforcement (\emph{Cavendish} [2015] UKSC 67 at [32]); the test tolerates supra-compensatory exposure, and the verified post-2015 span is informative: in \emph{Houssein v London Credit Ltd} the Court of Appeal corrected the test ([2024] EWCA Civ 721) and, on remission, a default rate at four times the standard rate was held non-penal ([2025] EWHC 2749 (Ch)), a result affirmed on a second appeal ([2026] EWCA Civ 830). We accordingly centre the English multiple at $m=3$ with a defensible range of $1.2$--$4$, against $m=1$ for the United States, where Restatement (Second) \S~356 treats unreasonably large agreed damages as unenforceable and evaluates reasonableness in light of the anticipated or actual loss and the difficulties of proof.

\paragraph{Why the English multiple is centred at three.} Reading \emph{Cavendish} as tolerating only modest supra-compensatory exposure would put the English multiple near unity. The \emph{Houssein} chain shows the tolerance to be substantially wider, and no point value in $[1.2,4]$ is defensible as a measurement; we therefore centre at $3$ and carry the endpoints as sensitivity cases. At the centre the upper margin is solvency in the market occupations except software development. This does not imply full separation, because the litigation threshold independently creates a lower pool. All British point values below are scenarios, not estimates of English doctrine.

\paragraph{Why the Austrian and Swiss multiples are finite.} It is tempting to extend the German merchant benchmark across the civil law, and the extension is wrong in both directions. The Austrian merchant rule was repealed with effect from 2007, and the Swiss Code of Obligations never contained one; treating either as unbounded imports a rule that has no referent in that jurisdiction. The numerical reach of the distinction is bounded in closed form (Corollary~\ref{cor:mx}): every cell whose multiple exceeds the stated threshold is invariant to the value the multiple takes, which covers the Swiss radiology cell over the full band and three of four market occupations above $1.667$. The unbounded benchmark survives only where it belongs---in the German negotiated channel and in Austrian contracts formed before the reform.

\emph{(iii) The cost rule.} Enforcement is viable only if the party who must fund the claim expects to recover at least what it spends; at exact indifference the claimant files (A6). Let $\rho_p$ denote the probability of prevailing. Under a loser-pays rule with a degressive statutory fee schedule---the German GKG/RVG structure---the claimant's expected cost is $(1-\rho_p)\bigl(c_0+c_1\Lam^{a_c}\bigr)$, and enforcement is viable iff $\rho_p\Lam\geq(1-\rho_p)(c_0+c_1\Lam^{a_c})$. Under the English rule with a small-claims track, costs are not recoverable below a threshold, which produces a step rather than a smooth boundary. Under the American rule each side bears its own costs, and the claim is brought if either the client's own-cost calculation or a contingency-funded attorney's participation constraint is satisfied. The parameter $\kappa$ captures collective redress: class aggregation multiplies the effective stake, which is why the American rule excludes small claims severely in B2B---where arbitration clauses suppress aggregation---and much less severely in B2C.

Given \eqref{eq:wedge}, the expected liability cost of a provider of type $s$ is
\begin{equation}\label{eq:cost}
C(s,K)\;=\;\theta\,\pi\,\bigl(1-\rho(s)\bigr)\,\Lam(K;J,v).
\end{equation}

Table~\ref{tab:oper} operationalizes the four primitives soberly: the legal source, the model object it maps to, the calibrated value, and the observable that would identify it empirically. Nothing in the table is estimated; the mapping is the claim.

\begin{table}[t]\centering\footnotesize
\caption{Operationalization of the legal primitives}
\label{tab:oper}
\begin{tabular}{p{2.2cm}p{2.4cm}p{4.3cm}p{2.1cm}p{3.2cm}}
\toprule
Primitive & Model object & Legal source (examples) & Calibrated value & Identifying observable\\
\midrule
Cost allocation & viability indicator; $c_0,c_1,a_c$, small-claims step & \S~91 ZPO with GKG/RVG tables (DE); CPR Small Claims Track (UK); American rule with contingency funding (US) & rule type and $(c_0,c_1,a_c)$, Table~\ref{tab:jur} & statutory fee schedules; filing rates by claim size\\
Penalty doctrine & cap $m\kappa v$ on $\Lam$ & \S\S~339~ff.\ BGB; \S~348 HGB, negotiated channel (DE); \S~307 case-group ceilings, standard-terms channel (DE); HaR\"AG BGBl I 2005/120, \S~1336(2) ABGB, \S~907(18) UGB (AT, vintage-\allowbreak indexed); art.~163(3) OR, BGE 143 III 1 (CH); \emph{Cavendish} [2015] UKSC 67 [32]; \emph{Houssein} [2024] EWCA Civ 721, [2025] EWHC 2749 (Ch), affirmed [2026] EWCA Civ 830 (UK, central 3, range 1.2--4); Restatement (2d) \S~356 (US) & channel/\allowbreak vintage-\allowbreak indexed, Table~\ref{tab:jur} & liquidated-damages multiples in contract corpora; vintage discontinuity in Austrian contracts\\
Displacement / pooling & factor $(1-\psi)$ & Art.~34 GG, \S~839 BGB; AHG (AT); VG art.~3(3) (CH); Crown Proceedings Act 1947; NHS Resolution CNST & $\psi$, Table~\ref{tab:jur} & recourse incidence; premium differentiation in indemnity schemes\\
Mandatory minimum & lowest enforceable cap $f\,v$, standard-terms channel only & \S\S~305~ff., 307 BGB, no \emph{geltungserhaltende Reduktion} (DE); \S\S~879(3), 864a ABGB (AT); art.~100 OR (CH) & $f\in\{0.60,0.45,$ $0.20,0\}$ & case law voiding caps; surveyed cap levels\\
Collective redress & stake multiplier $\kappa$ & FRCP Rule 23 (US); Directive (EU) 2020/1828 (B2C only) & $\kappa\in\{1,8\}$ & class-action incidence by claim size\\
\bottomrule
\end{tabular}
\end{table}

\subsection{Timing and equilibrium}

Within a period: (1) the firm chooses $h$, which determines next period's skill through \eqref{eq:skill}; (2) the firm posts the cap $K$; (3) clients observe $K$ and a public record of verified rescues and failures, form beliefs $\mu(s\mid L,\text{record})$, and choose a provider; (4) the AI fails with probability $\pi$, the human rescues with probability $\rho(s)$, failures are verified with probability $\theta$, and verified failures trigger payment of $\Lam$. The separating pricing environment, posterior formation and off-path beliefs of \citet{bauer2026b} are maintained primitives throughout: this paper does not re-solve the buyer's problem, and its contribution is to characterize which portions of that separating allocation remain legally implementable under the wedge, in the least-cost sense of \citet{riley1979}.

\subsection{Provenance and relation to the companion papers}\label{sec:provenance}

\ifanon Because this paper builds directly on two recent working papers, we state component provenance in the main text rather than in an appendix.\else Because this paper stands on two companion papers by the same author, we state component provenance in the main text rather than in an appendix.\fi Table~\ref{tab:provenance} lists every model component, its source, and its status here. No result of \citet{bauer2026a} or \citet{bauer2026b} is re-derived or re-claimed; the contribution of this paper is confined to the rows marked new.

\begin{table}[h]\centering\small
\caption{Provenance of model components}
\label{tab:provenance}
\begin{tabular}{p{5.6cm}p{3.6cm}p{3.6cm}}
\toprule
Component & Source & Status in this paper\\
\midrule
Skill transition \eqref{eq:skill} including its non-negativity operator, learning/erosion rates & \citet{singh2026} & Used unchanged\\
Rescue technology \eqref{eq:rho}, schedule \eqref{eq:schedule}, solvency truncation & \citet{bauer2026b} & Used unchanged\\
Outcome-contingent signal class, $\theta\Lam\geq v$, retained-risk principle & \citet{bauer2026a} & Used unchanged\\
Institutional wedge \eqref{eq:wedge} & --- & New\\
Two-sided window, mandatory-minimum floor (Prop.~\ref{prop:bunch}, Cor.~\ref{cor:floor}) & --- & New\\
Penalty-doctrine condition (Prop.~\ref{prop:penalty}) & --- & New\\
Stake--ceiling compression (Prop.~\ref{prop:payer}) & --- & New\\
Verifiability-investment direction (Prop.~\ref{prop:theta}) & --- & New\\
Reliability paradox (Prop.~\ref{prop:paradox}) & \citet{singh2026} & Restated; only the institutional interaction is new\\
Channel and vintage indexation of the jurisdiction vectors; control-mode result (Prop.~\ref{prop:mode}) & --- & New\\
Penalty-cap irrelevance and vintage monotonicity (Cor.~\ref{cor:mx}, Cor.~\ref{cor:vintage}) & --- & New\\
Occupation $\times$ jurisdiction calibration; provider-scale result; diagnosis map & --- & New\\
\bottomrule
\end{tabular}
\end{table}

\begin{remark}[On the notation]\label{rem:notation}
\citet{bauer2026a} writes verifiability as $\varphi$ and the liability promise as $D$.
\citet{bauer2026b} writes them as $\theta$ and $L$, because $\varphi$ is the learning
rate of the skill dynamics it takes from \citet{singh2026}.  We follow
\citet{bauer2026b} throughout: $\theta$ is verifiability, $L$ a posted cap, $\varphi$
the learning rate in \eqref{eq:skill}, and $D$ the loss on a verified failure in (A5).
Where the institutional wedge maps a posted cap into an effective recovery we write
$\Lam$ and reserve $L$ for the face amount.  A reader of either companion paper can
therefore map every symbol without ambiguity.
\end{remark}

\section{The institutional signaling window}\label{sec:window}

\subsection{Deductive architecture and standing assumptions}\label{sec:architecture}

Three kinds of statement are kept apart throughout, and the reader is entitled to know which is which. \emph{Inherited primitives} are results of the companion papers, used without re-derivation: the rescue technology, the least-cost separating schedule~\eqref{eq:schedule}, the client stage with its posterior formation and off-path beliefs, and the solvency truncation. \emph{Derived results}---Lemma~\ref{lem:impl}, Propositions~\ref{prop:sc}--\ref{prop:paradox} and Corollaries~\ref{cor:floor}--\ref{cor:vintage}---follow from those primitives together with the standing assumptions below, and from nothing else; in particular no proposition depends on the calibration. \emph{Calibrated magnitudes} (Section~\ref{sec:calib}) supply orders of magnitude for the comparative statics and are tested for directional robustness in Section~\ref{sec:robust}. The legal sources in Appendix~C identify mechanisms and orderings; they are not measurements, and we never treat them as such.

The standing assumptions are institutional rather than technical, and each is a claim about positive law that can fail:

\begin{itemize}
\item[(A1)] \textbf{Unlimited default.} The statutory default is unlimited fault-based liability, so the provider's only contractual margin is how far to limit it. The posted object is a cap subtracted from that default, not an increment added to a finite one.
\item[(A2)] \textbf{Destructive control (void-snap).} Where mandatory law strikes a cap, the clause falls away in its entirety and the dispositive default fills the gap; there is no reduction to the permissible maximum. This is what truncates the enforceable message space, and it is derived from the doctrine rather than assumed about buyer inference.
\item[(A3)] \textbf{Independence of control modes.} Ceiling control (judicial moderation of agreed damages) and floor control (nullity of caps) operate as separate institutional channels; courts do not substitute one for the other. Where they do, the two margins become substitutes and Lemma~\ref{lem:impl} requires a weaker statement (Remark~\ref{rem:scope}).
\item[(A4)] \textbf{Channel exclusivity.} A single posted term is either individually negotiated or pre-formulated, never both, so the unbounded merchant multiple and the standard-terms floor cannot attach to the same clause.
\item[(A5)] \textbf{Deterministic loss and non-degeneracy.} The loss on a verified failure is the value at stake $v$, so a cap $K\leq v$ pays $\min\{D,K\}=K$ and gross exposure is the face amount $K+A$. With a stochastic loss the same statements hold for expected payments with $K$ replaced by $\mathbb{E}\min\{D,K\}$, which weakens the floor's bite; this is the sharpest untested modelling assumption in the paper. We also require $\chi>0$ and $\theta\pi>0$: where the client's stake vanishes, no belief is worth buying and no separating equilibrium exists in any regime.
\item[(A6)] \textbf{Verification and enforcement.} A failure is attributed and proven with probability $\theta$, and a proven claim is brought where the claimant expects to recover at least the claim costs under the applicable cost rule. At exact indifference the claimant files. This tie break makes every participation component closed and its lower endpoint admissible; under the opposite tie break the endpoint would be an infimum and the boundary equilibria below would require a separate non-attainment analysis.
\item[(A7)] \textbf{Topology for the interval equilibrium.} The positive image of the institutional wedge has one connected component. Lemma~\ref{lem:msg} does not assume this; Propositions~\ref{prop:tax}--\ref{prop:bunch} do. Cells with more than one positive component are either covered by the inactive-gap case of Remark~\ref{rem:holes} or reported outside the proved domain.
\end{itemize}

Assumptions (A1), (A2) and (A4) are statements of German and Austrian standard-terms doctrine and are stated for that channel; (A3) is an empirical claim about adjudication, examined in companion work; (A5) is the standard enforcement friction of the litigation literature \citep{shavell1982}. Where an assumption fails, the corresponding margin disappears rather than moving---which is exactly how the British and American cells behave in Table~\ref{tab:heatmap}.

\subsection{Single crossing under the wedge}

The results below are stated in the space of enforceable exposure. Because the wedge \eqref{eq:wedge} is flat at its penalty kink and jumps at the enforcement boundary, it is not invertible; the following lemma guarantees that working in exposure space is nevertheless without loss.

\begin{lemma}[Instrument decomposition and implementation]\label{lem:impl}
A liability cap and an agreed-damages clause are distinct contractual carriers with distinct
control regimes, and the object the model prices is implemented by an observable \emph{pair} of
them. Let $K\in[fv,v]$ be the cap on compensatory damages and $A\ge0$ the agreed sum. Under (A5)
the loss on a verified failure is $D=v\ge K$, so the compensatory component pays
$\min\{D,K\}=K$. The gross commitment is
\begin{equation}\label{eq:paymap}
G \;=\; \Pi(K,A),\quad
\Pi(K,A)=\begin{cases} \max\{K,A\} & \text{default (non-cumulation)},\\[2pt]
K+A & \text{where cumulation is expressly agreed and enforceable.}\end{cases}
\end{equation}
The default branch is the positive law of every civil-law jurisdiction we calibrate: the
conventional penalty counts as a \emph{minimum} of damages and is set off against them, with only
the excess recoverable in addition ($\S$~340(2), 341(2) BGB; $\S$~1336(3) ABGB since HaR\"AG
2005, which is dispositive; art.~161(2) OR).

Assume $f\le1\le m\kappa$ and let the penalty doctrine cap the agreed sum, $A\le m\kappa v$. The
two branches then reach \emph{different} sets of gross commitments,
\[
\Pi_{\max}:\ [\,fv,\;m\kappa v\,],
\qquad
\Pi_{\mathrm{sum}}:\ [\,fv,\;v+m\kappa v\,],
\]
so the cumulation branch reaches strictly further: at $v=1$, $f=1/2$, $m\kappa=1$ the admissible
pair $(K,A)=(1,1)$ gives a face amount of $2$. What is invariant is the \emph{retained exposure},
because the wedge \eqref{eq:wedge} truncates at $m\kappa v$:
\[
\bigl\{\Lam(\Pi(K,A);J,v):(K,A)\ \text{admissible}\bigr\}
=\bigl[(1-\psi)fv,\;(1-\psi)m\kappa v\bigr]
\quad\text{under \emph{both} branches,}
\]
and nothing below $(1-\psi)fv$ is implementable---in cells where enforcement is viable on all of
the admissible domain. Where it is not, the displayed interval is replaced by the exact image of
Lemma~\ref{lem:msg}, which carries the enforcement indicator of \eqref{eq:wedge} explicitly; what
the present lemma asserts survives that replacement, because $\Lam$ depends on the pair only
through $G$ and the indicator is therefore applied identically on both branches. Since the signal
is carried by retained exposure and not by the face amount, the window is invariant to the
cumulation rule; only the implementing pair, and the contractual amount it names, change.
\end{lemma}

\begin{proof}
Take $G\in[fv,v]$. Both branches implement it by $(K,A)=(G,0)$, admissible since $G\ge fv$. Take
$G\in(v,m\kappa v]$. The cumulation branch implements it by $(v,G-v)$ and the default branch by
$(v,G)$, since $\max\{v,G\}=G$; in both cases $K=v\ge fv$ because $f\le1$, and
$A\le m\kappa v$ respects the doctrinal cap because $m\kappa\ge1$. Conversely $K\ge fv$ and
$A\ge0$ give $\Pi(K,A)\ge K\ge fv$ under both branches, so no $G<fv$ is implementable: an
attempted cap below $fv$ is void in its entirety and the dispositive default applies (A2),
returning the provider to unlimited exposure, which is dominated for every type. If $f>m\kappa$
the interval is empty and no clause is admissible.

Since both the provider's cost and the client's inference depend on $(K,A)$ only through $G$,
all implementing pairs are payoff-equivalent and the minimal implementing pair is a selection
without economic content. Statements in exposure space are therefore without loss, subject to the
enforcement-viability indicator of \eqref{eq:wedge}: where enforcement is non-viable, retained
exposure is zero whatever pair is posted, and the lemma is vacuous for that cell.
\end{proof}

\begin{remark}[What the decomposition buys, and what it costs]\label{rem:decomp}
Three things. First, it answers the objection that a single scalar cannot play both the role of a
cap subtracted from an unlimited default and the role of a supra-compensatory add-on: it cannot,
and it does not have to. Second, it locates the two control modes on the two carriers rather than
on one clause---the floor binds $K$, the ceiling binds $G$---so
Proposition~\ref{prop:mode} follows from the carriers instead of being asserted. Third, and least
comfortable, it makes an observability demand explicit: because $(K,0)$ and $(K,A)$ differ in
gross commitment but share a cap, the timing must have clients observe the pair, not the cap
alone. We assume that, and note that it is an assumption about contract transparency rather than
a result.

The invariance in Lemma~\ref{lem:impl} is worth isolating. The two payment maps are legally very
different---one lets the creditor add the penalty to full damages, the other treats the penalty
as a floor on damages---yet they generate the same set of implementable gross commitments and
therefore the same window. Every numerical result below is unchanged if the cumulation branch is
replaced by the default branch throughout.
\end{remark}

\begin{proposition}[Survival of single crossing]\label{prop:sc}
Fix $(J,v)$. On any interval of posted caps over which $\Lam(\cdot;J,v)$ is strictly increasing, the cost function \eqref{eq:cost} satisfies
\[
\frac{\partial C}{\partial \Lam}=\theta\pi\bigl(1-\rho(s)\bigr)>0,
\qquad
\frac{\partial^2 C}{\partial s\,\partial \Lam}=-\theta\pi\rho'(s)<0 ,
\]
so the Spence--Mirrlees condition holds on that interval; conditional on the pricing environment inherited from \citet{bauer2026b}---the client stage, posterior formation and the Riley construction---the least-cost separating equilibrium of that environment extends to it. Conversely, if $\Lam(\cdot;J,v)$ is constant on the feasible range of posted caps---which occurs whenever $\psi=1$, or $m\kappa v\leq\underline{L}$, or enforcement is non-viable at every feasible $\Lam$---then $C$ does not vary with $K$ and no separating equilibrium in the cap exists.
\end{proposition}

\begin{proof}
See Appendix~A.1. The first two displays are immediate from \eqref{eq:cost} and $\rho'>0$; the converse is the observation that a signal whose cost is type-independent is uninformative in every PBE, which is the argument of \citet[Prop.~5, type-independent certification]{bauer2026a} applied to the wedge rather than to provenance certification.
\end{proof}

Proposition~\ref{prop:sc} identifies three qualitatively distinct institutional failure modes, and it is worth being clear that they are not substitutes for one another. Displacement ($\psi\to1$) removes the provider from the liability relation entirely. The penalty cap removes the \emph{range} over which the retained exposure can vary. Non-viable enforcement removes the client's incentive to invoke it. A reform that addresses one leaves the other two intact.

\subsection{Two-sided compression}
\subsubsection*{The actual message space}

Write $x=\Lam(K;J,v)$ for retained exposure and $u(s)=1-\rho(s)>0$.
Conditional on receiver belief $b$, type $s$ obtains the signaling component
\begin{equation}\label{eq:signalpayoff}
U(s,x;b)=\chi b-\theta\pi u(s)x.
\end{equation}
Let the legally admissible posted domain be
$\mathcal D=[fv,G]$, where
$G=\min\{\bar L_{\rm solv},\bar L_{\rm law},m\kappa v\}$, and take
$\mathcal D=\varnothing$ when $fv>G$.  Require
$v>0$, $0\le f$, $0\le\psi\le1$, $m,\kappa>0$, and non-negative enforcement
costs.  Write $\mathcal V\subseteq\mathcal D$ for the posted amounts satisfying
the claimant's participation constraint.  No monotonicity of $\mathcal V$ is
assumed.

\begin{lemma}[Retained message set]\label{lem:msg}
For every occupation--jurisdiction cell, write the positive components of
$\mathcal V$ as the disjoint intervals
$[E_j,H_j]$, $j=1,\ldots,q$.  Then the exact retained message set is
\begin{equation}\label{eq:messageset}
\mathcal X=
\begin{cases}
\varnothing, & \mathcal D=\varnothing,\\[2pt]
\bigl(\{0\}\ \text{if }\mathcal D\setminus\mathcal V\ne\varnothing
\text{ or }\psi=1\bigr)
\ \cup\ \displaystyle\bigcup_{j=1}^{q}[(1-\psi)E_j,(1-\psi)H_j],
& \mathcal D\ne\varnothing.
\end{cases}
\end{equation}
Thus $q=0$ is a structural-zero regime (with no allocation when
$\mathcal D=\varnothing$); $q=1$ gives either $[F,C]$ or
$\{0\}\cup[F,C]$; and $q\ge2$ gives a holed positive message set.  The
single-interval equilibrium below is asserted only for $q=1$, except for the
separately proved inactive-gap case in Remark~\ref{rem:holes}.
\end{lemma}

\begin{proof}
On $\mathcal D$ the penalty, statutory and solvency ceilings do not bind again,
so a viable posted amount $K$ maps to $(1-\psi)K$ and a non-viable amount maps
to zero.  A positive affine rescaling maps every component $[E_j,H_j]$ to the
displayed component and preserves every gap.  If $\psi=1$ all components
collapse to zero. If $\mathcal D=\varnothing$, there is no posted-clause
allocation. If $\mathcal D\ne\varnothing$ but $\mathcal V=\varnothing$, every
admissible posted clause maps to zero. These cases exhaust the image.
\end{proof}

For the English small-claims rule the below-threshold component begins at
$q_{\rm sc}c_0/\rho_p$ and ends at $T_{\rm sc}$, while the ordinary-cost component
begins at the positive solution of $K=c_0+c_1K^{a_c}$.  The image is holed
exactly when the latter solution exceeds $T_{\rm sc}$ and both components meet
$\mathcal D$.  Hence connectedness is a restriction, not a consequence of the
wedge.

Lemma~\ref{lem:msg} deliberately assumes nothing about the shape of
$\mathcal V$.  The comparative statics of Corollaries~\ref{cor:floor}--%
\ref{cor:vintage} need one property of it, and that property is not an
assumption but a consequence of the cost rules of \eqref{eq:wedge}.

\begin{lemma}[Upper-tail viability and the attainability of zero]\label{lem:uppertail}
Under the loser-pays and American cost rules, with cost exponent
$a_c\in(0,1)$ and non-negative $c_0,c_1$, the viable set is an upper set:
$K\in\mathcal V$, $K'\in\mathcal D$ and $K'\ge K$ imply $K'\in\mathcal V$.
Hence $q\le1$, the positive component reaches the upper end of $\mathcal D$,
and zero is attainable iff the participation root exceeds the mandatory floor.
Under the loser-pays and American rules the attainability of zero therefore
does not depend on the penalty multiple: participation roots and the mandatory
floor are independent of $m$, $m$ moves only the upper end of $\mathcal D$, and
by the upper-set property the enlargement is viable.  Under the English
small-claims rule the upper-set property fails, the image may be holed, and the
independence fails with it: where a mandatory floor sits below the small-claims
threshold, raising $m$ past $T_{\rm sc}/(\kappa v)$ pushes the ceiling into a
non-viable region and makes zero attainable, which moves the lower cutoff
(Remark~\ref{rem:threshold}).  Zero-attainability is $m$-invariant in every
configuration calibrated here only because the calibrated English cells carry
no mandatory floor.
\end{lemma}

\begin{proof}
Write viability as $h(K)\ge0$ with
$h(K)=\rho_pK-c\,(c_0+c_1K^{a_c})$, where $c=1-\rho_p$ under loser-pays and
$c=1$ or the attorney cost share under the two American routes.  Then
$h(0)=-c\,c_0\le0$,
$h''(K)=c\,c_1a_c(1-a_c)K^{a_c-2}>0$ and $h'(K)\to\rho_p>0$, so $h$ is
strictly convex, has at most one positive root, and is strictly positive above
it; the viable set is therefore an upper set and its lower edge is that root,
which does not involve $m$.  For the English rule viability below $T_{\rm sc}$
requires $\rho_pK\ge q_{\rm sc}c_0$ and above it
$\rho_pK\ge(1-\rho_p)(c_0+c_1K^{a_c})$; when the second root exceeds
$T_{\rm sc}$ the interval between them is non-viable, which is the holed case.
Both roots and $T_{\rm sc}$ are free of $m$.  Finally, zero is attainable iff
some admissible posted amount lies outside every component; raising $m$ enlarges
$\mathcal D$ upwards only, and by the upper-set property the enlargement is
viable, so the answer is unchanged.  In the configurations calibrated here the
English rule carries no mandatory floor, so its least viable amount is strictly
positive and zero is attainable at every multiple.
\end{proof}

The last sentence is a property of the calibrated configurations, not of the
rule, and the combination it excludes is the one case in which an upper-margin
reform has a lower-margin effect.

\begin{remark}[Fee-shifting threshold with a mandatory floor]\label{rem:threshold}
Take the English small-claims rule with $f\bar v<T_{\rm sc}$, a below-threshold
root at or below the floor, and an ordinary-cost root above $T_{\rm sc}$.  Then
every admissible amount is viable while $\kappa m\bar v\le T_{\rm sc}$, so zero
is unattainable and separation is anchored at $s_F$; as soon as
\[
 m>\frac{T_{\rm sc}}{\kappa\bar v},
\]
the ceiling clears the threshold into a region where no claimant participates,
zero becomes attainable, and the lower cutoff jumps to the interior root of
\eqref{eq:lowercutoff}.  Raising the penalty cap then shrinks the separating set
\emph{from below}.  On the British cost vector with $c_0$ at $1.5$ times its
benchmark, a floor of $0.60$ and the legal-research ticket, the switch occurs at
$m^{\dagger}=0.800$ and moves the cutoff from $s_F=0.030$ to $0.326$; the
ordinary-cost component reappears at $m=0.934$, after which the cell is a
declared hole.  No calibrated cell combines the two features---the three
civil-law systems shift costs at every stake and the two British channels carry
no floor---so the diagnosis map is unaffected.  Whether a legal system pairs a
fee-shifting threshold with a statutory minimum liability is an empirical
question we leave open.
\end{remark}

The distinction matters because a candidate that places the bottom type at $F$
even in the second case is not an equilibrium.  That type could retain
the same on-path belief by moving to the limit of the zero-exposure region and
save $\theta\pi u(s_F)F>0$.  The deviation is strict for every admissible
$F>0$ and is not removed by the continuum above $F$.

\begin{proposition}[Endogenous lower pool and separating schedule]\label{prop:tax}
Assume $q=1$, $0<F\le C$, a uniform type distribution on $[s_F,\alpha]$, $\chi>0$,
$\theta\pi>0$, $\rho'>0$, and $\rho(\alpha)<1$.  Define
\[
 \rho_L(a)=\E[\rho(S)\mid S\le a],\qquad
 Q(a)=\frac{\chi[\rho(a)-\rho_L(a)]}{\theta\pi u(a)}.
\]
If the message set is $[F,C]$, put $a=s_F$.  If it is
$\{0\}\cup[F,C]$ and $0<F<Q(\alpha)$, there is a unique $a\in(s_F,\alpha)$
satisfying
\begin{equation}\label{eq:lowercutoff}
F=Q(a).
\end{equation}
Types $s<a$ pool at zero, and the least-cost separating schedule beginning at
$a$ is
\begin{equation}\label{eq:lamF}
\Lam_a(s)=F+\frac{\chi}{\theta\pi}
 \log\!\frac{u(a)}{u(s)},\qquad s\ge a.
\end{equation}
If the message set is $\{0\}\cup[F,C]$ and $F\ge Q(\alpha)$, an all-zero PBE is
supported by the D1 belief $\rho(\alpha)$ after the least positive message. For
the power rescue
technology \eqref{eq:rho}, $Q$ is strictly increasing, so these cases are
exhaustive conditional on $q=1$. If $q=0$, Propositions~\ref{prop:tax}--
\ref{prop:bunch} assert no allocation: the cell is reported only as structural
zero (no positive retained message), not as pooling at a positive ceiling.
\end{proposition}

\begin{proof}
Step 1 (cutoff indifference).  A cutoff type obtains
$\chi\rho_L(a)$ at zero and $\chi\rho(a)-\theta\pi u(a)F$ at the first
positive message.  Equality is exactly \eqref{eq:lowercutoff}.

Step 2 (local IC).  The first-order condition for truthful reporting is
$\Lam'(s)=\chi\rho'(s)/(\theta\pi u(s))$.  Integrating from $a$ and imposing
$\Lam(a)=F$ gives \eqref{eq:lamF}; no inherited constant $c_0$ enters.

Step 3 (pool types).  A type $s<a$ that mimics separating type $t\ge a$ gains
\[
G(s,t)=\chi[\rho(t)-\rho_L(a)]-\theta\pi u(s)\Lam_a(t).
\]
Using the differential equation,
$\partial_tG=\chi\rho'(t)[1-u(s)/u(t)]<0$ for $s<t$.  Its best deviation is
therefore $t=a$.  At $a$ the cutoff is indifferent by Step 1; since
$u(s)>u(a)$, every $s<a$ strictly prefers zero.

Step 4 (separating types versus zero).  Let
$V(s)=U(s,\Lam_a(s);\rho(s))-\chi\rho_L(a)$.  Step 1 gives $V(a)=0$, while
the envelope theorem gives $V'(s)=\theta\pi\rho'(s)\Lam_a(s)>0$.
Thus all $s>a$ strictly reject zero.  Global IC among positive messages follows
from $r-1-\log r\ge0$, exactly as in the unconstrained Riley schedule.

Step 5 (existence and uniqueness).  Continuity gives $Q(s_F)=0$.  For
$\rho(s)=q(s/\alpha)^p$, direct differentiation of the closed-form mean
\[
\rho_L(a)=\frac{q}{(p+1)\alpha^p}
 \frac{a^{p+1}-s_F^{p+1}}{a-s_F}
\]
yields $Q'(a)>0$ for $a>s_F$; the numerator reduces to positive multiples of
$1-y^p$ and $p(1-y)-y(1-y^p)$ with $y=s_F/a\in(0,1)$.
The intermediate-value theorem proves the unique interior root when
$F<Q(\alpha)$. In the zero-admissible case, if $F\ge Q(\alpha)$ even the top
type weakly rejects the least costly positive message under the D1 belief
$\rho(\alpha)$; every larger message costs more, so all types pool at zero.
\end{proof}

\begin{remark}[Participation and the actions outside the schedule]\label{rem:participation}
Two actions outside the separating schedule have to be priced, and the answer
differs by branch.  Where the zero message is admissible, the cutoff $a$ is
defined by the indifference of type $a$ between it and separation, and the
envelope $V'(s)=\theta\pi\rho'(s)\Lam_a(s)>0$ of Step 4 makes the rent strictly
positive above $a$.  Participation is therefore implied by membership of the
separating set and imposes no separate restriction: we verify this numerically
over the calibrated cells and the ticket grid.  Where the zero message is not
admissible---$\mathcal X$ has no zero, or $F=0$---not posting is not an action,
the cheapest admissible message is $F$, it is the message of the bottom type,
and the same envelope argument applies among admissible messages.  A fixed cost
of posting would fall on positive messages and not on the zero message; it
would therefore have to enter the cutoff condition as
$F+k_L/(\theta\pi u(a))=Q(a)$ as well as the comparison above, both being the
same comparison made by the same type.  We normalise it to zero throughout and
report the sensitivity to a positive value in Section~\ref{sec:robust}.

What the model does not contain is a market-exit option: the provider is
assumed to serve the client, and the outside option of declining the engagement
is not priced.  The client-side fringe is normalised to zero surplus, following
\citet{bauer2026b}, so all premia are measured against buying the unassisted
machine; the provider-side counterpart is absent, and every statement below is
conditional on participation in that sense.
\end{remark}

\begin{proposition}[Upper pool and three-region equilibrium]\label{prop:bunch}
Maintain Proposition~\ref{prop:tax} and suppose $a<\alpha$.  Define
$\rho_H(c)=\E[\rho(S)\mid S\ge c]$ and
\begin{equation}\label{eq:phi}
\Phi_a(c)=C-\Lam_a(c)-
 \frac{\chi[\rho_H(c)-\rho(c)]}{\theta\pi u(c)}.
\end{equation}
Then $\Phi_a'(c)<0$.  If $\Lam_a(\alpha)\le C$, types in $[a,\alpha]$
separate.  If $\Phi_a(a)>0>\Phi_a(\alpha)$, there is a unique
$c^*\in(a,\alpha)$; types below $a$ pool at zero, types in $[a,c^*]$
separate according to \eqref{eq:lamF}, and types above $c^*$ pool at $C$.
If $\Phi_a(a)\le0$, pooling at $C$ is a PBE only if the separate zero-message
deviation inequality
\begin{equation}\label{eq:emptyzero}
\max_{s\in[s_F,\alpha]}\{\chi[\rho(s_F)-\E\rho(S)]
  +\theta\pi u(s)C\}\le0
\end{equation}
also holds whenever zero is admissible.  The belief $\rho(s_F)$ is the D1
belief; $0$ is not an admissible posterior mean because
$\rho(S)\ge\rho(s_F)>0$ almost surely.
\end{proposition}

\begin{proof}
For a uniform distribution,
$\rho_H'(c)=[\rho_H(c)-\rho(c)]/(\alpha-c)$.  Differentiating
\eqref{eq:phi} and substituting \eqref{eq:lamF} gives
\[
\Phi_a'(c)=-\frac{\chi[\rho_H(c)-\rho(c)]}{\theta\pi u(c)^2}
 \left[\rho'(c)+\frac{u(c)}{\alpha-c}\right]<0.
\]
Continuity establishes the stated trichotomy.  The marginal type's equality is
\eqref{eq:phi}; the usual single-crossing derivative rules out every lower
separator's deviation into the upper pool.  For unused messages in
$(\Lam_a(c^*),C)$ assign the D1 belief $\rho(c^*)$.  The D1 threshold is
decreasing below $c^*$ and increasing above it, so its unique minimum is the
cutoff.  Pool types' best on-path deviation is likewise $c^*$; separators'
best gap deviation is no better than mimicking $c^*$.  Finally, a low-pool
type's best positive separating deviation is $a$ by Proposition~\ref{prop:tax};
and its deviation to $C$ is no better than the cutoff's upper-pool option by
single crossing.  Equation \eqref{eq:emptyzero} is logically independent when
the separating interval disappears, which is why it is stated rather than
silently inferred from $\Phi_a(a)\le0$.
\end{proof}

\begin{remark}[Disconnected positive images]\label{rem:holes}
When $q\ge2$, Propositions~\ref{prop:tax}--\ref{prop:bunch} do not apply merely
by replacing the union with its convex hull.  One special case is immediate:
if the candidate is fully separating and
$\Lam_a(\alpha)\le(1-\psi)H_1$, the entire path lies below the first hole.
Assign the top-type D1 belief $\rho(\alpha)$ to every higher component.  Those
messages cost more than the top type's on-path message and cannot attract any
lower type by single crossing, so the full-separation equilibrium survives.
If the path meets a hole, this paper supplies no equilibrium and reports the
cell as outside the proved domain rather than interpolating across the hole.
\end{remark}

\begin{corollary}[Institutional signaling window]\label{cor:floor}
The equilibrium separating set is $[a,c^*]$, with either endpoint allowed to
coincide with $s_F$ or $\alpha$.  Under Lemma~\ref{lem:uppertail}, increasing
$C$ weakly raises $c^*$ and never changes $a$.  Without that property the
statement fails: if raising the ceiling exposes admissible amounts on which
enforcement is not viable, those amounts map to the zero message, zero becomes
attainable, and $a$ moves to the interior root of \eqref{eq:lowercutoff}---an
upper-margin change with a lower-margin effect.  Changing a statutory floor has two distinct effects.  While it
remains below the enforcement threshold, it changes neither $F$ nor the
equilibrium.  Once it eliminates the zero-exposure region, the lower pool
disappears ($a$ falls to $s_F$), but the positive schedule is re-anchored at the
higher $F$, which can lower $c^*$.  Hence no global monotonic sign for the
separating share in the statutory floor follows without an additional regime
restriction.
\end{corollary}

\begin{remark}[Why the ceiling is an attractor, not a truncation]\label{rem:attract}
The upper result is unchanged by the lower repair: a marginal type gains the
upper pool's discrete belief premium at vanishing cost.  Therefore a reserved
gap separates the last separating message from $C$.
\end{remark}

\begin{remark}[Scope of the floor result]\label{rem:scope}
The lower pool is caused by the interaction of a discontinuous litigation
threshold and a legally admissible commitment below that threshold.  If law
forbids all such commitments, the message set is connected and separation
starts at $s_F$.  The result is thus conditional on the exact image of the
institutional wedge, not on the label ``floor.''
\end{remark}

\begin{proposition}[Which side a mandatory rule constrains]\label{prop:mode}
The side on which mandatory control binds is determined not by the remedy---destructive
invalidity versus judicial moderation---but by the pair \emph{(carrier, direction of the
excess)}. A rule that strikes or reduces a term because it is too \emph{low} constrains from
below and generates a floor; a rule that strikes or reduces a term because it is too \emph{high}
constrains from above and generates a ceiling. The remedy determines only whether the constraint
is attained with a jump (destructive invalidity, where the offending term falls away and the
default applies) or continuously (moderation, where the term survives at a reduced level).
Applied to the two carriers of Lemma~\ref{lem:impl}: standard-terms control voids a cap $K$
beneath the contract-typical foreseeable loss---too low, hence a floor---while the penalty
doctrine strikes or reduces agreed damages $A$ that are excessive---too high, hence a ceiling.
A two-sided \emph{mandatory-control} wedge therefore requires both carriers to be regulated, which in the calibrated sample is the standard-terms channel. This is a classification of positive law, not a theorem of the model, and it does not say that a two-sided \emph{equilibrium window} needs two mandatory rules: an enforcement floor together with a solvency ceiling produces one as well, which is exactly how the British cells behave.
\end{proposition}

\begin{remark}[Scope of the legal mapping]\label{rem:mode}
The proposition classifies the four mandatory rules the model prices; it does not claim to
exhaust B2B mandatory control. Partial invalidity, severance, reduction-to-validity in systems
that permit it, and residual good-faith review all exist and can move a constraint from one side
to the other or make it non-binding. We treat each (carrier, channel, remedy, default) combination
as a primitive to be read off positive law rather than inferred from the remedy label, and the
numerical multiples are stylized scenario values, not identified legal quantities.
\end{remark}

\subsection{The penalty doctrine and the provability gross-up}

\begin{proposition}[Deterrence-adequate recovery under the penalty doctrine]\label{prop:penalty}
Deterrence-adequate recovery requires a gross commitment satisfying $\theta G\geq v$: the loss must be grossed up by the inverse of the probability that fault can be established \citep[Prop.~3, the $\varphi D\geq v$ condition]{bauer2026a}. Under a penalty doctrine that caps enforceable agreed damages at $m$ times compensatory loss, two distinct conditions must be kept apart. The gross payment $G=\min\{K+A,\,m\kappa v\}$, the retained exposure $R=(1-\psi)G$ that carries the signal, and the client's expected recovery $\theta G$ are different objects, and each governs a different requirement. \emph{Deterrence-adequate recovery} requires $\theta\bar G\ge v$, where $\bar G=\min\{m\kappa v,\ \bar L_{\mathrm{solv}},\ \bar L_{\mathrm{law}}\}$ is the smallest binding \emph{gross} ceiling. Isolating the penalty carrier gives the necessary sub-condition
\begin{equation}\label{eq:penaltycond}
m\;\geq\;\frac{1}{\theta\,\kappa},
\end{equation}
which is sufficient only where every other gross ceiling is slack---not in general: with $v=1$, $\theta=1/2$, $\kappa=1$, $m=3$ and $\bar L_{\mathrm{solv}}=1$ the sub-condition holds ($3\ge2$) while $\theta\bar G=1/2<v$. It is also \emph{not necessary} under the cumulation branch of Lemma~\ref{lem:impl}, where the two carriers add and the maximum gross commitment is $v+m\kappa v$ rather than $m\kappa v$: the correct threshold there is $m\ge(1/\theta-1)/\kappa$, and at $v=1$, $\theta=1/2$, $\kappa=1$, $K=A=1$ recovery is adequate ($\theta G=1\ge v$) although $m=1<2=1/(\theta\kappa)$. Equation~\eqref{eq:penaltycond} is therefore the non-cumulation branch, and the branch must be declared before the threshold is used. In it
the displacement factor does \emph{not} appear: what an insurer, pool or state pays on the provider's behalf still reaches the victim. \emph{Signal-adequate retained exposure} requires $R$ to vary enough to separate, which is the window condition of Proposition~\ref{prop:bunch}, and there $(1-\psi)$ enters in full. Conflating the two---requiring $m\ge1/[\theta(1-\psi)\kappa]$ for recovery---overstates the doctrinal constraint whenever $\psi>0$: at $\theta=1/2$, $\kappa=2$, $\psi=1/2$ and $m=2$ the conflated bound is $1/[\theta(1-\psi)\kappa]=2$, so it binds with equality while recovery $\theta\kappa mv=v$ is already adequate. In the private cells with $\psi=0$ and $\kappa=1$, \eqref{eq:penaltycond} reads $m\ge1/\theta$ and the two conditions coincide. The American B2C channel is the exception among the cells reported below: there $\kappa=8$, so \eqref{eq:penaltycond} reads $m\ge1/(8\theta)$ and holds in all five occupations at $m=1$, while the recovery condition $\theta\bar G\ge v$ continues to fail in legal research, radiology and social work because the solvency ceiling binds first.
Condition \eqref{eq:penaltycond} governs recovery, not separation: by Proposition~\ref{prop:bunch} a continuum separating equilibrium can exist while \eqref{eq:penaltycond} fails, whenever $\Lam_a(\alpha)\le C$ so that the whole schedule fits below the ceiling. If \eqref{eq:penaltycond} fails, the firm can restore deterrence-adequate recovery only by raising $\theta$ to at least $1/(m\kappa)$---an operational investment in verifiability---or by shifting to an instrument outside the doctrine.
\end{proposition}

\begin{proof}
Appendix~A.2.
\end{proof}

\begin{remark}\label{rem:penalty}
Condition~\eqref{eq:penaltycond} governs the \emph{deterrence-adequate} exposure $\Lam\geq v/\theta$ inherited from the two-type benchmark of \citet{bauer2026a}. Below the threshold, a truncated window certifying only low types can persist, because the continuum least-cost schedule starts at $F$, which lies below $v/\theta$ in every calibrated cell; Corollary~\ref{cor:mx} gives the multiple above which the window itself is invariant.
\end{remark}

Proposition~\ref{prop:penalty} has a clean empirical content, and it is not a statement about legal families. In the German negotiated channel $m=\infty$ and the condition never binds; in the German standard-terms channel the case-group ceiling $m\approx1.05$ lies below every calibrated $1/\theta$, so the gross-up is contractually unavailable in every market occupation and the response must be operational even in Germany. At the central Austrian post-reform multiple the condition fails exactly in legal research ($1/\theta=2.00>1.5$); over the Swiss band it fails in legal research below $m=2.0$ and holds everywhere at the upper end. In the United Kingdom, at the central multiple, it is satisfied in the four market occupations---at the lower endpoint $m=1.2$ the count falls to one (audit)---and in the United States it binds at every calibrated verifiability; social work, with $1/\theta=4$, fails everywhere. Table~\ref{tab:penalty} reports the condition across the calibrated configurations. The dividing line runs between configurations ($m\gtrless1/\theta$), not between families; the common-law cells are simply where the calibrated multiples sit lowest. The mechanism---a doctrine of contract law determining whether a service firm's optimal quality-assurance response is legal or operational---is to our knowledge stated formally here for the first time.

\begin{corollary}[Penalty-cap irrelevance]\label{cor:mx}
Because $\rho(\alpha)=\rho_{\max}<1$, the required exposure is bounded: $\Lam_a(\alpha;J)=F(J)+\frac{\chi}{\theta_0\pi}\ln\frac{1-\rho(a)}{1-\rho_{\max}}<\infty$. The penalty cap therefore ceases to affect the \emph{window}---as opposed to the ceiling---at
\[
m^{\times}(\mathrm{occ},J)=\frac{\min\{(1-\psi)\bar L_{\mathrm{solv}},\;(1-\psi)\bar L_{\mathrm{law}},\;\Lam_a(\alpha;J)\}}{(1-\psi)\,\kappa\bar v},\qquad \psi<1.
\]
For every $m\ge m^{\times}$ the cell coincides with the unbounded-multiple
benchmark, provided the attainability of the zero message does not itself move
with $m$---which Lemma~\ref{lem:uppertail} establishes for every configuration
calibrated here, and which must be checked before the threshold is applied to a
cost rule outside that class. With the schedule anchored at $a$, the software thresholds are
$1.405$ (German standard terms), $1.405$ (Austrian standard terms), $1.417$
(Swiss standard terms) and $1.426$ (United Kingdom). The solvency plateaus
remain $1.600$ in legal research, $1.333$ in radiology and $1.667$ in audit.
These thresholds are regenerated after solving \eqref{eq:lowercutoff}; anchoring
the schedule at the inherited $s_F$ instead gives different values.
\end{corollary}

\begin{corollary}[Vintage monotonicity]\label{cor:vintage}
Fix an occupation and a channel. Across the Austrian reform the two regimes differ only in the penalty multiple, hence only in the ceiling: $\bar\Lam^{\mathrm{pre}}\ge\bar\Lam^{\mathrm{post}}$, while the floor is regime-invariant (the mandatory fraction does not depend on vintage, and by Lemma~\ref{lem:uppertail} neither the enforcement floor nor the attainability of the zero message moves with the multiple). Since the lower edge of the window depends only on the floor and the upper edge is monotone in the ceiling, $[\sfl,\sbar]^{\mathrm{post}}\subseteq[\sfl,\sbar]^{\mathrm{pre}}$: the reform weakly narrows the window in every occupation. Strictness requires more than $m^{\mathrm{post}}<m^{\times}$: the pre-reform window must be non-empty and the cutoff must actually move, since $c^*$ solves the nonlinear equation \eqref{eq:phi} and both regimes may already be empty. No proportionality between the compression and the shortfall $m^{\times}-m^{\mathrm{post}}$ is claimed. The corollary yields a sign prediction the calibration does not itself test: the compression is ordered by how far $m^{\mathrm{post}}$ falls below $m^{\times}$. Within a cell the set inclusion is exact; \emph{across} cells no ordering follows from the corollary, because $c^*$ solves the nonlinear equation \eqref{eq:phi} and cells differ in $H$, in scale and in slope, so a larger shortfall can produce a smaller type compression. The cross-occupation ranking below is therefore a calibration-dependent hypothesis, not a consequence of the corollary: on the calibrated vectors it is largest in audit, moderate in legal research, and exactly zero in software development and radiology---in software because the post-reform centre $m^{\mathrm{post}}=1.5$ already exceeds the cell's invariance threshold $m^{\times}$ of Corollary~\ref{cor:mx}, in radiology for any post-reform multiple above $1.333$ (Section~\ref{sec:calib}; Figure~\ref{fig:vintage}).
\end{corollary}

\begin{remark}[Exposure posted as a multiple of the fee]\label{rem:multiple}
Business-service contracts post a cap as a multiple $k$ of the fee rather than
as an amount, so the message is $k$ and the exposure is $k\,p(\hat s)$, with $p$
the fee at the posted belief.  Single crossing is unaffected:
$\partial^2C/\partial s\,\partial k=-\theta\pi\rho'(s)p(\hat s)<0$.  Local
incentive compatibility returns \eqref{eq:lamF} unchanged, so the multiple is
$k^*(s)=\Lam_a(s)/p(s)$ and the reparametrisation is an isomorphism wherever
$k^*$ is injective.  Differentiating, and using $p'=\chi\rho'$ under either fee
rule,
\[
k^{*\prime}(s)>0
\quad\Longleftrightarrow\quad
p(s)>\theta\pi u(s)\Lam_a(s).
\]
On the net-of-pool fee $p=\chi[\rho-\rho_L(a)]$ this is the participation
comparison of Remark~\ref{rem:participation}, with equality at $a$ and strict
inequality above: the multiple is informative exactly where posting is
individually rational.  Where the zero message is inaccessible the cutoff is
$a=s_F$ instead, the comparison need not hold there, and $k^*$ falls before it
rises---in six calibrated cells, all Austrian standard terms, so the
reparametrisation is not injective near $s_F$.

A convention that bounds the multiple is a different object.  It imposes
$\Lam_a(s)\le\bar k\,p(s)$, a ceiling rising in the type, which truncates the
separating set from above at $k^*(s)=\bar k$ and empties it if and only if
$\bar k\le k^*(a)=F/p(a)$---the mandatory floor over the marginal type's fee.
Across the calibration $k^*(a)$ on the gross-fee rule runs from $2.4$ to $31.5$,
and from $2.4$ to $6.5$ in the cells of Figure~\ref{fig:corridor}, so a
convention at one or two times fees admits no separating type in any of them.  With $k$ confined to a finite grid, full separation fails on cardinality
and the equilibrium is a partition into at most $|K|$ cells.  What the model
adds to the observed distribution (Section~\ref{sec:conclusion}) is the
threshold $\bar k^{\dagger}=F/p(a)$, rising in the floor and falling in the fee.
The levels are orders of magnitude rather than measurements: the model's fee
$\chi\rho(s)$ carries the failure probability $\pi$ and a contractual fee does
not.
\end{remark}

\begin{table}[h]\centering\small
\caption{The penalty-doctrine condition $m\geq1/\theta$. UK evaluated at the central multiple $m=3$; entries in parentheses give the status at the lower endpoint $m=1.2$ of the defensible range.}
\label{tab:penalty}
\resizebox{\textwidth}{!}{%
\begin{tabular}{lcccccc}
\toprule
Occupation & $1/\theta$ & DE neg.\ ($\infty$) & DE st.\ ($1.05$) & AT post ($1.5$) & CH $[1.5,3.0]$ & UK ($m=3$)\\
\midrule
Software / web development & 1.25 & holds & fails & holds (fails at 1.2) & holds & holds (fails at 1.2)\\
Legal research \& advisory & 2.00 & holds & fails & \textbf{fails} (holds at 2.0) & fails below 2.0 & holds (fails at 1.2)\\
Radiology / imaging & 1.43 & holds & fails & holds (fails at 1.2) & holds & holds (fails at 1.2)\\
Audit \& assurance & 1.18 & holds & fails & holds & holds & holds (holds at 1.2)\\
Social work / child protection & 4.00 & holds & fails & fails & fails & fails\\
\bottomrule
\end{tabular}}

\smallskip
{\footnotesize The American B2B channel ($m=1.0$, $\kappa=1$) fails in every occupation at every calibrated $\theta<1$; for the B2C channel ($\kappa=8$) see Proposition~\ref{prop:penalty}.}
\end{table}

\subsection{The third-payer wedge}

In many expert-service markets the party who pays is not the beneficiary.  Put
$\chi=\omega\zeta v\pi$.
\begin{proposition}[Stake effects are two-margin effects]\label{prop:payer}
Where the zero message is inaccessible ($a=s_F$), increasing $\omega$ lowers
the upper cutoff whenever it is interior.  Where it is accessible,
\[
F=\frac{\omega\zeta v}{\theta}
 \frac{\rho(a)-\rho_L(a)}{u(a)}
\quad\Longrightarrow\quad \frac{da}{d\omega}<0.
\]
Thus a larger stake expands separation from below but compresses it from above;
the sign of $d(c^*-a)/d\omega$ is not determined without parameter restrictions.
At $\omega=0$ beliefs have no value and separation is impossible.
\end{proposition}
\begin{proof}
The derivative of the lower equation has the stated sign because $Q'(a)>0$.
For the upper cutoff, $\Phi_a'<0$ throughout, so the implicit function theorem
gives
\[
\frac{dc^*}{d\omega}
=-\frac{1}{\Phi_a'(c^*)}
\left[\frac{\partial\Phi_a}{\partial\omega}
+\frac{\partial\Phi_a}{\partial a}\frac{da}{d\omega}\right].
\]
Where the zero message is inaccessible, $a=s_F$ is fixed and the bracket reduces to its first
term, which is negative.  Where it is accessible the second term does not
vanish, and it reinforces rather than offsets the first: by \eqref{eq:lamF},
$\partial\Lam_a(c)/\partial a=-(\chi/\theta\pi)\rho'(a)/u(a)<0$, so
$\partial\Phi_a/\partial a>0$, while $da/d\omega<0$ from the lower equation;
the product is negative and adds to $\partial\Phi_a/\partial\omega<0$.  Hence
$dc^*/d\omega<0$ in both cases, and the endogenous fall in $a$ deepens the
upper compression instead of undoing it.  The \emph{width} $c^*-a$ is the
difference of two falling boundaries and is therefore unsigned in general.
\end{proof}

\section{A disciplined calibration exercise}\label{sec:calib}

\subsection{Approach}

We follow the discipline of the companion papers: parameters are identified in sign and order of magnitude, not estimated. Where a published estimate exists we use it; where it does not, we place the parameter ordinally and say so. Every value carries a source in Appendix~C. The numeraire is $v=1\equiv\text{EUR }5{,}000$, the median freelance web-development project.

Two caveats govern the whole exercise. First, several of the sharpest empirical anchors are contested. The colonoscopy deskilling study \citep{budzyn2025} attracted peer-reviewed correspondence disputing its causal reading and carries a corrigendum restoring an omitted covariate; we therefore use it to place $\gamma$ ordinally high in radiology rather than to pin a value. Second, industry sources (code-churn telemetry, insurance capacity announcements) are not peer reviewed and are used only where peer-reviewed alternatives do not exist. We report results as regime classifications and directional comparisons, and we test in Section~\ref{sec:robust} whether those survive perturbation.

\subsection{Occupation vectors}

Table~\ref{tab:occ} reports the five occupation vectors. Software development anchors $\pi$ on the security-flaw literature: \citet{pearce2022} find 39.3\% of top Copilot suggestions vulnerable across 1{,}689 generated programs, \citet{fu2023} find 35.8\% of snippets in live repositories CWE-bearing, and a subsequent replication reports 27.3\%; we set $\pi=0.30$. Its erosion rate is the highest in the sample, following the employment evidence of \citet{brynjolfsson2025canaries}, who find a 16\% relative employment decline for workers aged 22--25 in exposed occupations and roughly 20\% for software developers---the on-ramp tasks through which juniors historically built fallback capability are precisely those the AI now performs. Verifiability is highest here after audit: git blame, test suites, reproducible builds and formal acceptance under a \emph{Werkvertrag} make defects attributable.

Legal research takes $\pi=0.25$ from \citet{magesh2025}, whose preregistered evaluation finds hallucination rates of 17\% for Lexis+ AI and 33\% for Westlaw AI-Assisted Research, against 58--88\% for general-purpose models \citep{dahl2024}. Its verifiability is only middling: citations are checkable, but the causal attribution of an adverse outcome to advice is contested.

Radiology takes a low failure rate ($\pi=0.10$) and a high erosion rate, and is the occupation where the human check is most compromised by the machine: \citet{dratsch2023} find that the share of correctly rated mammograms among inexperienced readers falls from 79.7\% to 19.8\% when the AI suggestion is wrong.

Social work is the break case. Its AI capability frontier is low because the core service is relational; its failure rate where AI is used is high \citep{moore2025}; its verifiability is the lowest in the sample, since formal documentation is extensive but forensic causal attribution is weak and diffused across agencies. Decisively, its value at stake is an order of magnitude above what an individual provider or small non-profit can bear: harm is decoupled from the fee, which is set per \emph{Fachleistungsstunde} at roughly EUR 55.

\subsection{Jurisdiction vectors and the diagnosis map}

Table~\ref{tab:jur} reports the jurisdiction vectors. Figure~\ref{fig:heatmap} gives the resulting diagnosis map: the share of the type space that separates in each occupation--jurisdiction cell; Figure~\ref{fig:window} draws the software equilibrium itself for three configurations.

\begin{table}[t]\centering\footnotesize
\caption{Occupation vectors. Numeraire $v=1\equiv$ EUR 5{,}000. Sources in Appendix~C.2.}
\label{tab:occ}
\begin{tabular}{lccccccccc}
\toprule
Occupation & $\alpha$ & $\pi$ & $\varphi$ & $\gamma$ & $\theta_0$ & $\rho_{\max}$ & $\bar v$ & $\bar L_{\text{solv}}$ & $h_{\min}$\\
\midrule
Software / web development & 0.88 & 0.30 & 0.60 & 0.55 & 0.80 & 0.85 & 1.0 & 3.0 & 0.753\\
Legal research \& advisory & 0.82 & 0.25 & 0.45 & 0.35 & 0.50 & 0.90 & 2.5 & 4.0 & 0.757\\
Radiology / imaging & 0.90 & 0.10 & 0.40 & 0.45 & 0.70 & 0.88 & 3.0 & 4.0 & 0.918\\
Audit \& assurance & 0.85 & 0.20 & 0.50 & 0.30 & 0.85 & 0.90 & 3.0 & 5.0 & 0.750\\
Social work / child protection & 0.35 & 0.35 & 0.20 & 0.30 & 0.25 & 0.60 & 8.0 & 0.5 & 0.811\\
\bottomrule
\end{tabular}
\end{table}

\begin{table}[t]\centering\footnotesize
\caption{Jurisdiction vectors. $m$: enforceable multiple of loss; $\psi$: displacement; $f$: cap floor as a fraction of $v$---an illustrative benchmark consistent with ordinal legal severity, not a measured quantity (Appendix~C.3); $\kappa$: collective-redress aggregation.}
\label{tab:jur}
\begin{tabular}{lcccccc}
\toprule
Jurisdiction & Cost rule & $c_0$ & $m$ & $\psi_{\text{public}}$ & $f$ & $\kappa$\\
\midrule
Germany --- negotiated & loser pays (degressive) & 0.40 & $\infty$\textsuperscript{$\ddagger$} & 0.95 & 0 & 1\\
Germany --- standard terms & loser pays (degressive) & 0.40 & 1.05\textsuperscript{\S} & 0.95 & 0.60 & 1\\
Austria --- contracts $<$ 2007 & loser pays (degressive) & 0.38 & $\infty$ & 0.95 & 0\,/\,0.45\textsuperscript{$\P$} & 1\\
Austria --- contracts $\geq$ 2007 & loser pays (degressive) & 0.38 & 1.5\textsuperscript{$\sharp$} & 0.95 & 0\,/\,0.45\textsuperscript{$\P$} & 1\\
Switzerland & loser pays & 0.55 & $[1.5,3.0]$\textsuperscript{$\flat$} & 1.00 & 0\,/\,0.20\textsuperscript{$\P$} & 1\\
United Kingdom & English rule + small claims & 0.90 & 3\textsuperscript{$\dagger$} & 0.15 & 0.00 & 1\\
UK (NHS channel) & English rule + small claims & 0.90 & 3\textsuperscript{$\dagger$} & 0.90 & 0.00 & 1\\
United States (B2B) & American rule + contingency & 1.50 & 1.00 & 0.55 & 0.00 & 1\\
United States (B2C) & American rule + contingency & 1.50 & 1.00 & 0.55 & 0.00 & 8\\
\bottomrule
\end{tabular}

\smallskip
{\footnotesize $\dagger$\,Central value; the defensible English range is $1.2$--$4$ (\emph{Cavendish} [2015] UKSC 67; \emph{Houssein} [2024] EWCA Civ 721, [2025] EWHC 2749 (Ch), affirmed [2026] EWCA Civ 830). No point value in the range is defensible as a measurement; see Section~\ref{sec:model} and the band sensitivity in Section~\ref{sec:robust}. $\ddagger$\,Stylized; \S~242 BGB caps extreme merchant penalties (BGH I ZR 168/05), a boundary we do not quantify; the negotiated-channel results are conditional on it lying above the solvency ceiling. \S\,Case-group encoding of total enforceable exposure, band $[1.00,1.30]$ (Section~\ref{sec:wedge}). $\sharp$\,Central value, band $1.2$--$2.0$; \S~1336(2) ABGB, vintage per \S~907(18) UGB. $\flat$\,Band result; no calibrated point value exists (Appendix~C.3). $\P$\,Floors attach to the standard-terms channel only (\S\S~879(3), 864a ABGB; art.~100 OR); negotiated-channel floor is 0.}
\end{table}

\begin{figure}[t]\centering
\includegraphics[width=\textwidth]{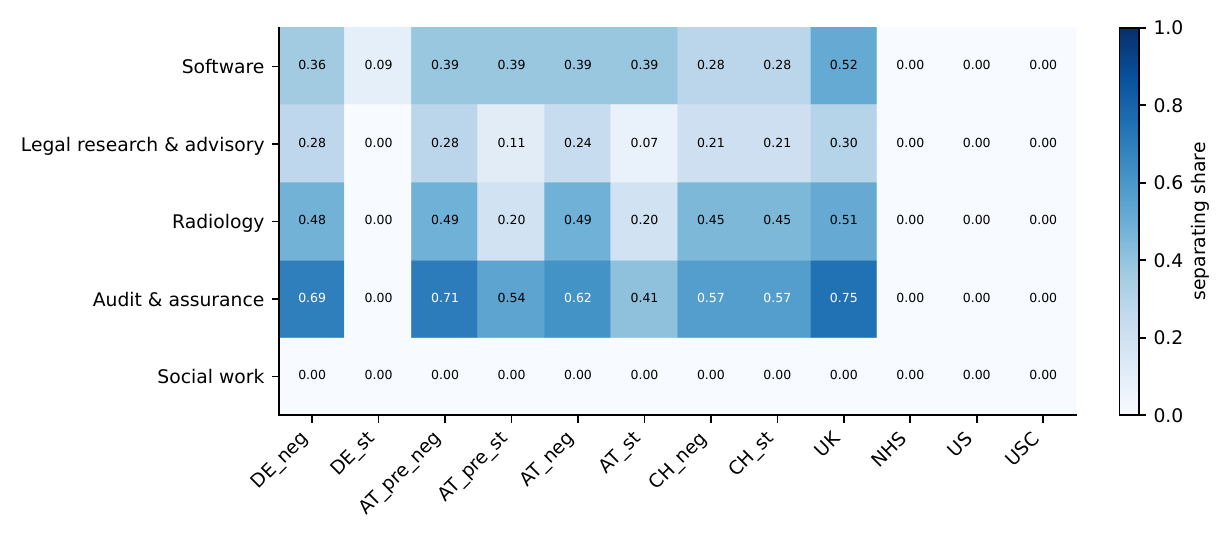}
\caption{Diagnosis map at the central multiples: equilibrium separating share by occupation and contracting channel. Zero cells may reflect enforcement, ceilings, displacement or their interaction; the figure does not identify a causal legal effect.}
\label{fig:heatmap}
\end{figure}

\begin{figure}[t]\centering
\includegraphics[width=0.8\textwidth]{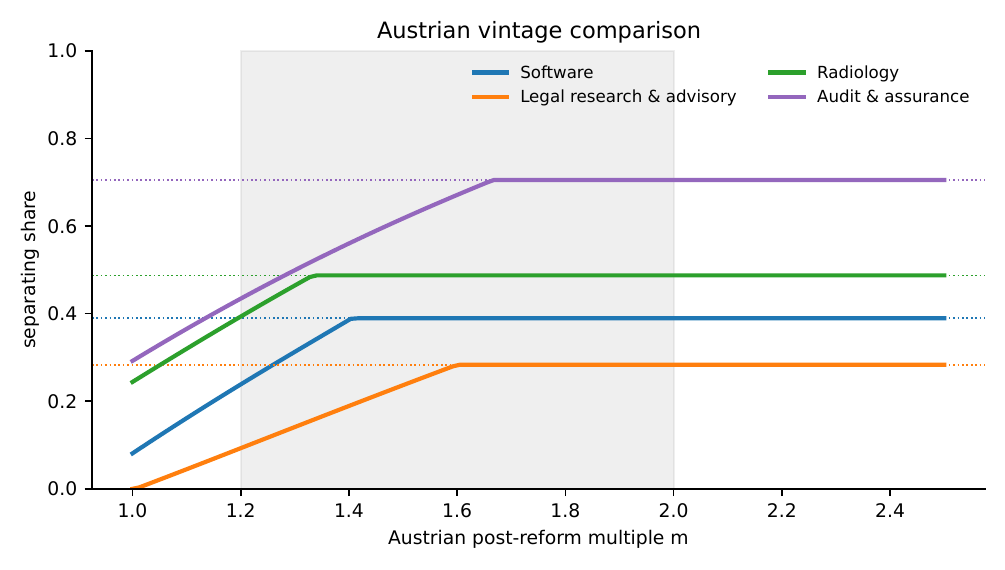}
\caption{The Austrian discontinuity. Equilibrium share in the negotiated channel as a function of the post-reform multiple; dotted horizontal lines give the pre-2007 level ($m=\infty$), and shading gives the scenario band $[1.2,2.0]$.}
\label{fig:vintage}
\end{figure}

\begin{table}[t]\centering\small
\caption{Share of the type space that separates in equilibrium, by occupation and institutional configuration. Austria post-reform at the central multiple $1.5$; Switzerland over the band $m\in[1.5,3.0]$, reported as a range where the band bites. Values are equilibrium separating shares in the sense of Proposition~\ref{prop:bunch}, not intersections of the schedule with the legal interval.}
\label{tab:heatmap}
\resizebox{\textwidth}{!}{%
\begin{tabular}{lccccccccc}
\toprule
& DE neg & DE st & AT pre & AT pre & AT post & AT post & CH & CH & UK\\
& & & (neg) & (st) & (neg) & (st) & (neg) & (st) & \\
\midrule
Software / web dev. & 0.360 & 0.092 & 0.389 & 0.389 & 0.389 & 0.389 & 0.283 & 0.283 & 0.516\\
Legal research & 0.277 & 0.000 & 0.283 & 0.113 & 0.237 & 0.070 & 0.212--0.258 & 0.212--0.258 & 0.303\\
Radiology / imaging & 0.479 & 0.000 & 0.487 & 0.195 & 0.487 & 0.195 & 0.453 & 0.453 & 0.514\\
Audit \& assurance & 0.693 & 0.000 & 0.705 & 0.541 & 0.617 & 0.413 & 0.570--0.656 & 0.570--0.656 & 0.749\\
Social work & 0.000 & 0.000 & 0.000 & 0.000 & 0.000 & 0.000 & 0.000 & 0.000 & 0.000\\
\bottomrule
\end{tabular}
}
\vspace{5pt}

\begin{tabular}{lccc}
\toprule
& UK NHS & US B2B & US B2C\\
\midrule
Software / web dev. & 0.000 (C-pool) & 0.000 (struct.) & 0.000 (struct.)\\
Legal research & 0.000 (C-pool) & 0.000 (struct.) & 0.000 (struct.)\\
Radiology / imaging & 0.000 (C-pool) & 0.000 (struct.) & 0.000 (struct.)\\
Audit \& assurance & 0.000 (C-pool) & 0.000 (struct.) & 0.000 (struct.)\\
Social work & 0.000 (C-pool) & 0.000 (struct.) & 0.000 (struct.)\\
\bottomrule
\end{tabular}
\vspace{2pt}

\footnotesize ``C-pool'' means that a positive interval $[F,C]$ exists but the
equilibrium pools at $C$; ``struct.'' means that no positive retained message
exists.  A structural-zero cell is not tested as a C-pooling allocation.
\end{table}

Three regularities stand out. All values are equilibrium separating shares
(Proposition~\ref{prop:bunch}). First, a litigation threshold creates substantial
lower pooling even in negotiated channels: for software the lower cutoff is
$0.574$ in Germany and $0.639$ in Switzerland.  This is why the corrected full
separation claims disappear. Second, standard-terms control need not eliminate
the signal: at the German central multiple $m=1.05$, software retains a small
window (share $0.092$), whereas the other German standard-terms cells have zero
share.  This is a point statement only.  Across the declared band
$m\in[1.00,1.30]$, software ranges from $0.051$ to $0.283$, legal research
remains zero, radiology opens at $m=1.172865$ and reaches $0.068$, and audit
opens at $m=1.108292$ and reaches $0.130$.  Third, no country-ordering
claim survives the equilibrium.  In several cells negotiated
and standard-terms shares coincide because the statutory floor remains below the
litigation threshold; the calibration therefore supports configuration-specific
comparisons only, not a general ranking by legal family or contracting channel.

\emph{Pooling kills the signal as thoroughly as immunity.} The NHS channel has
a positive admissible interval but pools at its ceiling in every occupation
despite the Crown Proceedings Act 1947, because what destroys the signal is the
absence of provider-specific premium differentiation, not immunity.  By
contrast, both United States configurations are structural-zero cells at the
central vectors: their zero shares do not describe a positive-ceiling pool.

\begin{table}[t]\centering\small
\caption{Social work: equilibrium separating share by provider scale. Negotiated channel; standard-terms channel in parentheses where non-zero.}
\label{tab:scale}
\begin{tabular}{lccccc}
\toprule
Provider scale ($\bar L_{\mathrm{solv}}$) & Germany & Austria & Switzerland & UK & US (B2B)\\
\midrule
Individual professional (0.5) & --- & --- & --- & --- & ---\\
Small firm / provider (2.0) & --- & --- & --- & --- & ---\\
Large organisation (15.0) & 0.928 (0.000) & 0.702 (0.250) & 0.705 (0.485) & 0.953 & 0.046\\
\bottomrule
\end{tabular}
\end{table}

\emph{Social work fails everywhere at individual and small-provider scale.} The row is empty in every configuration at $\bar L_{\mathrm{solv}}\leq2.0$. The pledge cannot scale with the harm; the break is over-determined by displacement, solvency and payer misalignment. Table~\ref{tab:scale} shows that the window opens only at organisational scale, and then only outside the German standard-terms channel.

\subsection{Exclusion bands in ticket value}

Proposition~\ref{prop:bunch} holds $v$ fixed. Varying $v$ traces out the set of transaction sizes for which a given type can be certified. The bound at each end is the ticket value at which the target type enters, and leaves, the equilibrium separating set: $v_{\min}$ is where the lower cutoff $a(v)$ falls to $s$, and $v_{\max}$ is where the upper cutoff $c^{*}(v)$ does. Provider participation imposes no further restriction here, by Remark~\ref{rem:participation}: in the lower-pool branch---which is the branch every cell of Table~\ref{tab:band} is in---the cutoff is defined by the indifference of type $a$ between the zero message and separation, and the rent is strictly positive above it. Earlier versions attributed $v_{\min}$ to a separate participation condition carried over from \citet[Prop.~6, the exclusion band]{bauer2026b}. That condition was evaluated at the inherited anchor $c_0$ against the belief $\rho(s_F)$, neither of which is the equilibrium object here, and it bound in three of the eight cells; those three bounds are lower under the equilibrium condition. Table~\ref{tab:band} and Figure~\ref{fig:band} report the corrected bands.

\begin{table}[t]\centering\small
\caption{Exclusion band in transaction value for a representative high type ($s=0.6\alpha$), computed on the equilibrium separating set. Values are read from a 400-point logarithmic ticket grid and are therefore reported to grid precision.}
\label{tab:band}
\begin{tabular}{llrr}
\toprule
Occupation & Configuration & $v_{\min}$ (EUR) & $v_{\max}$ (EUR)\\
\midrule
Software / web dev. & Germany (negotiated) & 5{,}797 & 14{,}804\\
 & Germany (standard terms) & \multicolumn{2}{c}{no band}\\
 & United Kingdom & 3{,}915 & 14{,}173\\
 & United States (B2B) & \multicolumn{2}{c}{no band}\\
\midrule
Legal research & Germany (negotiated) & 3{,}080 & 10{,}444\\
 & United Kingdom & 2{,}036 & 10{,}219\\
\midrule
Radiology / imaging & Germany (negotiated) & 4{,}869 & 16{,}510\\
 & United Kingdom & 3{,}218 & 15{,}805\\
\midrule
Audit \& assurance & Germany (negotiated) & 5{,}198 & 21{,}920\\
 & United Kingdom & 3{,}511 & 20{,}984\\
\midrule
Social work & all & \multicolumn{2}{c}{no band}\\
\bottomrule
\end{tabular}
\end{table}

The German software-development band is EUR 5{,}797--14{,}804 and covers about
$25\%$ of the calibrated ticket distribution.  Two features of the
construction matter for reading the table.  Feasibility is evaluated on the
equilibrium separating set rather than on the pointwise intersection of the
unconstrained schedule with the institutional bounds, which would make the
bands much wider.  And the lower bound is the entry of the target type into
that set throughout: in the lower-pool branch participation is implied by
membership of the separating set (Remark~\ref{rem:participation}), so no
separate participation condition is imposed.

\begin{figure}[htbp]\centering
\includegraphics[width=\textwidth]{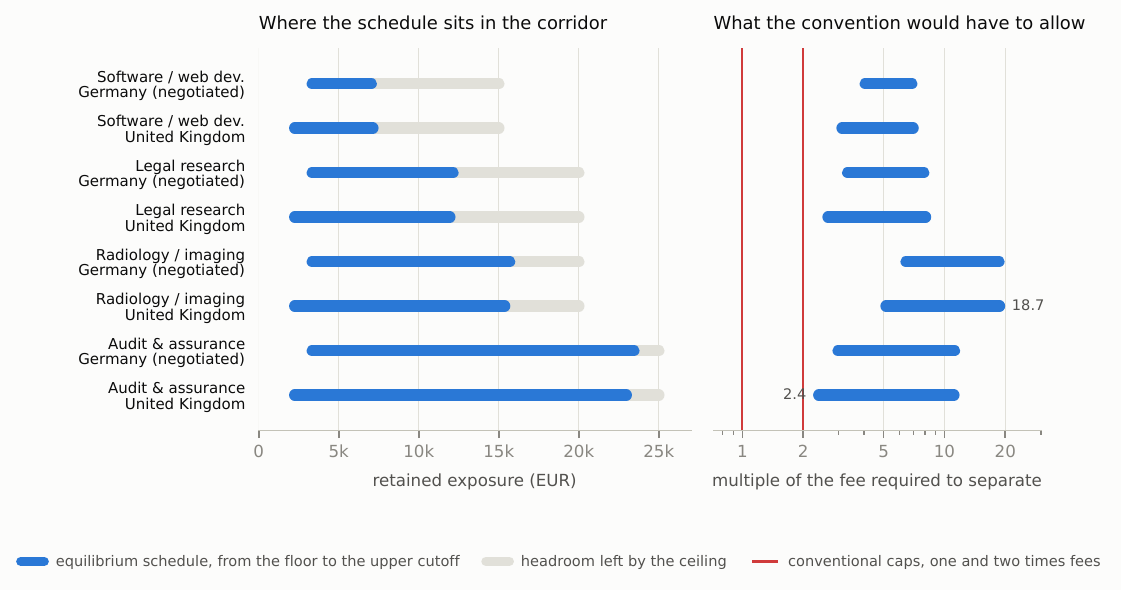}
\caption{The corridor and the multiple. Left: the schedule is anchored at the
mandatory floor $F$ and runs to $\Lam_a(c^*)$; the remainder to the binding
ceiling $C$ is headroom the equilibrium does not use. Right: the same schedule
expressed as a multiple of the fee, $k^*(s)=\Lam_a(s)/p(s)$ on the gross-fee
rule $p=\chi\rho$, against caps of one and two times fees. No calibrated cell
reaches either bound: the smallest requirement anywhere in the calibration is
$2.4$ times fees, in British audit; among the cells shown the largest is $18.7$,
in British radiology. Numeraire $v=1\equiv$ EUR
5{,}000; the levels carry the caveat of Remark~\ref{rem:multiple}.}
\label{fig:corridor}
\end{figure}

The American B2B band remains empty for the representative high type.  The
result is generated jointly by litigation viability and the penalty ceiling;
the table does not identify either institution as a causal effect.  British
and German bands overlap substantially at the central calibration, so the
exercise supports no ordinal country ranking in transaction coverage.

\begin{figure}[t]\centering
\includegraphics[width=0.86\textwidth]{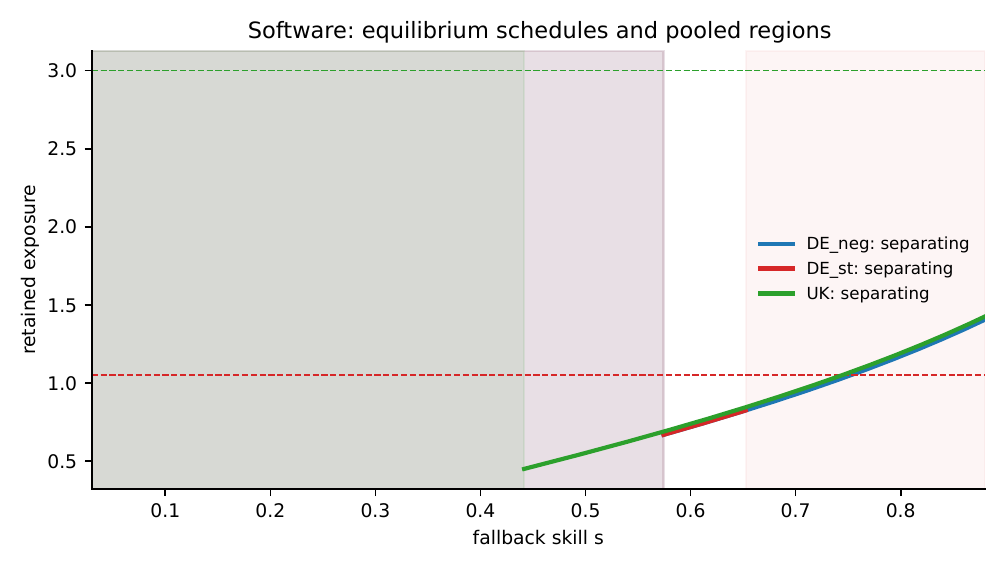}
\caption{Software equilibrium at the central calibration. Solid segments are separating schedules for the German negotiated, German standard-terms and UK configurations; vertical shading marks lower and upper pools, and dashed lines show the relevant ceilings. The figure displays the actual three-region equilibrium, not the intersection of an unconstrained schedule with legal bounds.}
\label{fig:window}
\end{figure}

\begin{figure}[t]\centering
\includegraphics[width=0.68\textwidth]{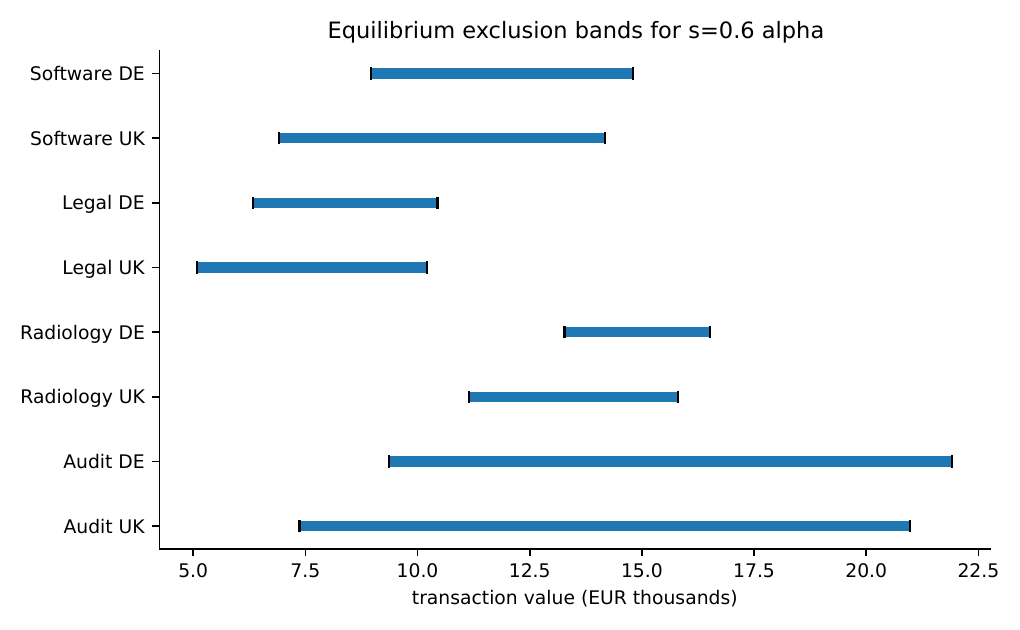}
\caption{Exclusion band in transaction value for a representative high type, against the ticket density.}
\label{fig:band}
\end{figure}

\section{Verifiability investment and the binding margin}\label{sec:theta}

Suppose workflow instrumentation raises $\theta$.  It reduces the exposure
needed to sustain any fixed belief difference, but it also makes the zero-cost
message relatively more attractive.  Instrumentation here acts on the
verifiability of \emph{performance}; the prevailing probability $\rho_p$ in
enforcement participation is a separate primitive and is held at its
calibrated value (Appendix~B), so the boundary $F$ does not move with
$\theta$.
\begin{proposition}[Verifiability has opposing boundary effects]\label{prop:theta}
Where the zero message is inaccessible the lower boundary is fixed at $s_F$,
and an interior upper cutoff rises with $\theta$.  Where it is accessible,
\[
F=\frac{\chi}{\theta\pi}
 \frac{\rho(a)-\rho_L(a)}{u(a)}
\quad\Longrightarrow\quad \frac{da}{d\theta}>0.
\]
Consequently the separating width need not be monotone in verifiability.  The
sign is positive only when the upper-boundary relaxation dominates the upward
movement of the lower cutoff; it is negative in a ceiling-slack cell.
\end{proposition}
\begin{proof}
Implicit differentiation of the lower equation gives $da/d\theta>0$ because
$Q'(a)>0$.  Conditional on $a$, differentiating \eqref{eq:phi} gives
$\partial\Phi_a/\partial\theta>0$, so the upper cutoff rises when interior.
If the ceiling is slack, $c^*=\alpha$ and only $a$ moves, proving the negative
case.  With both endpoints moving, no unconditional ordering follows.
\end{proof}

This qualification is economically important rather than cosmetic.  In the
German negotiated software cell the calibrated separating share rises from
$0.425$ at $\theta=0.30$ to $0.623$ at $0.40$, then falls to $0.312$ at
$0.90$ once the upper boundary is slack.  British legal research remains on
the upper-boundary-dominant branch over the same range.  Investment advice must
therefore be conditioned on the observed message topology and binding margin.

\section{Dynamics: engagement and the fallback stock}\label{sec:dynamics}

\subsection{The reliability paradox}\label{sec:paradox}

From \eqref{eq:skill}, skill is non-decreasing iff $\varphi\pi h\geq\gamma(1-h)$, giving a minimum engagement share
\begin{equation}\label{eq:hmin}
h_{\min}\;=\;\frac{\gamma}{\varphi\pi+\gamma}.
\end{equation}

\begin{proposition}[Reliability paradox; \citealp{singh2026}]\label{prop:paradox}
For an interior state $0<s<\alpha$ and positive rates $\varphi,\gamma>0$, $\pi\in(0,1)$ (the result does not extend continuously to $\pi=1$, where \eqref{eq:skill} multiplies the entire increment by $1-\pi=0$ and every $h$ holds $s$ fixed):
$\partial h_{\min}/\partial\pi<0$: as the AI becomes more reliable, the share of cases that must be worked by humans to hold fallback skill constant \emph{increases}, and $h_{\min}\to1$ as $\pi\to0$.
\end{proposition}

The proposition is immediate from \eqref{eq:hmin} but its content is not. Failures are the learning events; erosion runs on calendar time regardless. A more reliable AI supplies fewer learning events per unit of erosion, so the firm must manufacture engagement that the workflow no longer generates naturally. Our calibration yields $h_{\min}=0.75$ for software development and audit, $0.76$ for legal research, $0.81$ for social work, and $0.92$ for radiology (Figure~\ref{fig:paradox}). These levels are not delicate: halving or increasing $\gamma$ by half leaves $h_{\min}$ in $[0.60,0.82]$ for software development and in $[0.85,0.94]$ for radiology, so the qualitative claim---that maintaining unassisted radiological reading requires reserving the large majority of cases---survives the contested calibration of the erosion rate discussed in Section~\ref{sec:calib}. Radiology is the extreme case precisely because its AI is good: a department that wishes to preserve its radiologists' unassisted reading capability must have them read more than nine cases in ten without decisive machine input. This is an operational requirement, not an aspiration, and it is in direct tension with the throughput gains that motivate adoption.

\begin{figure}[t]\centering
\includegraphics[width=0.52\textwidth]{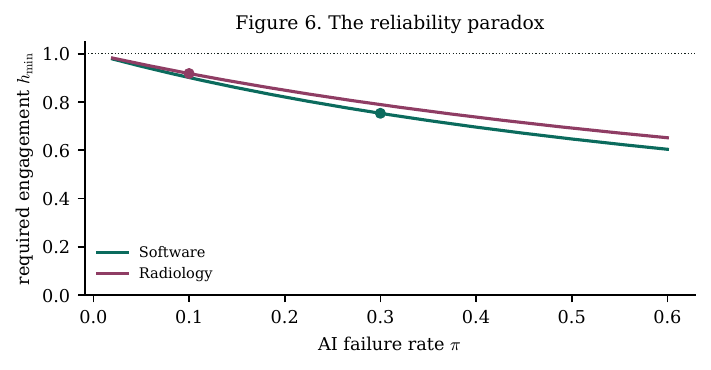}
\caption{The reliability paradox: required engagement $h_{\min}=\gamma/(\varphi\pi+\gamma)$ rises as the AI failure rate falls. Markers show the calibrated points.}
\label{fig:paradox}
\end{figure}

\begin{remark}[Learning from verification]\label{rem:verify}
Equation~\eqref{eq:skill} treats failures as the only learning events, following \citet{singh2026}. Suppose instead that actively checking a \emph{correct} AI output also builds skill, at rate $\varphi_v=\kappa\varphi$ with $\kappa\in[0,1]$, so that the drift becomes $[\varphi\pi+\varphi_v(1-\pi)]h-\gamma(1-h)$ and
\[
h_{\min}(\kappa)=\frac{\gamma}{\varphi\bigl(\kappa+\pi(1-\kappa)\bigr)+\gamma},
\qquad
\frac{\partial h_{\min}}{\partial\pi}\ \text{has the sign of}\ -(1-\kappa).
\]
The paradox therefore survives for \emph{every} $\kappa<1$ and disappears only in the knife-edge case where checking a correct output teaches exactly as much as diagnosing a failure. What changes is the level and, past a point, the ranking: at $\kappa=0.5$ the radiology requirement falls from $0.92$ to $0.67$, and radiology ceases to be the extreme case at $\kappa=0.39$, beyond which social work---whose low $\varphi$ limits the gain from verification---requires the most engagement. Evidence on automation bias places $\kappa$ near the bottom of its range rather than near one: \citet{dratsch2023} find that verification performed \emph{under} an AI suggestion degrades judgement, correct ratings among inexperienced readers falling from 79.7\% to 19.8\%. We report $\kappa=0$ as the baseline and treat $\kappa$ as the single most valuable parameter to estimate directly.
\end{remark}

\subsection{Why the dynamics stop here}\label{sec:dyn}

The natural next step is a market in which providers choose engagement, workers move, and
clients learn. We do not take it, and the reason is a result rather than a resource constraint.
Any such model must specify what a provider gains from a capable workforce when no contractual
window exists. If the objective values only a certifiable type, an empty window removes the
return to engagement by construction and the transition drives skill to the fringe: the
institutional dichotomy is then an implication of the decision rule, not a finding. If the
objective also values a public performance record---clients observe verified rescues and
failures whether or not a cap is posted---engagement survives the loss of the contractual
channel in every configuration we probed. The dichotomy is therefore an artefact of the first
specification and the substitution is the finding, which is the dynamic counterpart of the
reputation-substitution result of \citet{bauer2026b}.

Deciding between the two requires modelling the record: its censoring by imperfect verification,
its accumulation, and the weight buyers place on it against a posted commitment. That is a
different paper. What the present model does establish about dynamics is the boundary condition
of Section~\ref{sec:paradox}: the engagement share required to hold capability constant rises as
the machine improves, whatever the institutional regime. Where a record substitutes for
certification the regime governs the price of communication; where no record is informative---the
pooled-indemnity channel, in which no provider-specific outcome is attributed---it governs both.

\section{Robustness}\label{sec:robust}

Two claim codes recur below.  \textbf{R1} is the headline ranking claim: no
common-law configuration is strictly wider than the widest contemporary
civil-law configuration.  \textbf{R2a} is the ranking claim that the German
standard-terms channel is the narrowest of the six DACH configurations---a
statement about the ordering of separating shares, not about which primitive
binds at the bottom, since the mandatory floor ceases to bind once the ticket is
small enough that the enforcement floor overtakes it; the crossover is stated at
the calibrated $\rho_p=0.50$ and reverses for $\rho_p\geq0.520$.  Both are
evaluated only where a window exists.

The analysis separates theorem checks from calibration claims.
The verification harness evaluates eleven equilibrium or monotonicity
conditions and one message-topology classification in all $5\times12=60$
declared cells. No applicable condition fails. Verified/not
applicable counts are: strict decrease of $\Phi_a$ (34/26), no separator
deviation to the lower message (33/27), no low-pool mimicry (27/33), each of
the five interior upper-pool, gap and D1 conditions (25/35), direct
low-pool--to--upper-pool deviations (19/41), positive-boundary pooling (9/51),
structural-zero topology (18/42), and full separation (8/52). The last item but
one is a classification, not an equilibrium allocation: a structural-zero cell
is never counted as a verified positive-ceiling pool.  Those eighteen cells
carry two different states and we separate them.  In fourteen the enforceable
set is non-empty and no height in it is viable, so retained exposure is zero at
zero cost: the American B2B and B2C channels in all five occupations, and
social work in the four negotiated civil-law channels.  In the remaining four---social work
in the German, both Austrian and Swiss standard-terms channels---the enforceable
set is itself empty, $fv>\bar G$ by a factor between $3.2$ and $9.6$, so no
clause is admissible at all and by (A1)--(A2) the dispositive default applies at
the full exposure $(1-\psi)\min\{v,\bar L_{\rm solv},\bar L_{\rm law}\}$.  That
is a forced uninformative pool at maximal rather than zero exposure.  It does
not enter any claim count, and the two states should not be read as one. Equation~\eqref{eq:emptyzero} uses the D1 posterior \citep{chosobel1990}
$\rho(s_F)$; a belief value zero is outside the posterior support.  The program
returns ``not applicable'' rather than ``pass'' outside a condition's domain.

The headline ranking R1 does not survive perturbation: across the six band corners
it holds in only $5.7$--$9.1\%$ of the 88 occupation perturbations.  The legal
design contains 50 perturbations: the English small-claims
own-cost share $q_{\rm sc}=0.25$ and the US contingency-lawyer case-cost share
$q_{\rm att}=0.60$ are listed as primitives and each perturbed at
$\times0.5/\times1.5$, so the denominator is $50\times4=200$ claim-cells.
Four English-rule perturbations produce two positive message components.  In
three software cells the full schedule lies below the first hole and
Remark~\ref{rem:holes} verifies the equilibrium separately; in 13 cells the
candidate path meets a hole and no equilibrium is asserted.  Thus R1 holds in
$5.9$--$9.6\%$ of the 187 proved cells, with coverage $187/200=93.5\%$.
The Austrian vintage inclusion
remains exact in both designs; NHS pooling holds in $97.7\%$ of occupation
perturbations and all jurisdictional perturbations.  Other ordering claims are
diagnostic frequencies only and are not promoted to propositions.

A dedicated family varies the prevailing probability $\rho_p$ (held at $0.50$
throughout; Appendix~B) over $\{0.40,0.45,0.55,0.60\}$, cell by cell.  All
sixty central classifications are invariant on $[0.47,0.52]$, and 41 of 60 are
invariant on the full band.  Outside that window the family behaves exactly as
the cost-side comparative statics of Lemma~\ref{lem:uppertail} predict.
Downward, at $\rho_p\le 0.470$, the English-rule cells re-enter the hole
classes already declared under the $c_0$ family---the same unsupported-gap
cells, reached from the cost side---and one cell (NHS software) is
non-monotone, entering the declared hole class only on a sub-interval around
$0.45$.  Upward, lower participation roots open windows: the German
standard-terms software cell switches from interior separation to ceiling
pooling at $\rho_p\ge 0.520$ because the mandatory floor re-binds there; the
Austrian and German negotiated social-work cells open at $0.541$ and $0.554$;
and the US-B2C audit cell leaves structural zero at $\rho_p\ge 0.547$.  The
floor-regime statements in this section and in Appendix~C.3 are therefore
stated at the calibrated $\rho_p=0.50$, with these tipping values disclosed;
the claim counts above do not mix in this family.

One configuration is reported separately because it is a boundary of
Lemma~\ref{lem:uppertail} rather than a perturbation of a calibrated cell: a
fee-shifting threshold combined with a mandatory floor, described in
Remark~\ref{rem:threshold}.  There a pure ceiling change moves the lower
cutoff, at $m^{\dagger}=T_{\rm sc}/(\kappa\bar v)=0.800$ on the British cost
vector with $c_0$ at $1.5$ times its benchmark and a floor of $0.60$.  The
harness pins that switch; it enters no claim count, since no calibrated cell has
both features.

The regenerated Saltelli exercise \citep{saltelli2008} ($N=512$, 4{,}096 evaluations; re-run at
$N=2{,}048$ and under alternative seeds) finds a window in $39$--$41\%$ of
draws.  Displacement and the penalty multiple jointly dominate, with total
Sobol' indices between $0.44$ and $0.58$ across seeds; their ordering is not
resolved at these sample sizes and we do not assert one.  They are followed by
verifiability and the enforcement-cost block (total indices near $0.23$ each),
solvency ($0.15$--$0.16$) and the statutory fraction ($\le 0.04$).
First-order estimates for weak factors are noisy at this sample size and are
not interpreted by sign.  These are global sensitivity measures of a
calibrated model, not estimates of legal causal effects.

A separate family prices the normalisation of the fixed posting cost.  With
$k_L>0$ the lower-cutoff condition becomes $F+k_L/(\theta\pi u(a))=Q(a)$, since
the cutoff type pays the fixed cost on the positive message and not on the
zero message.  At $k_L=0.02$---two per cent of the numeraire ticket---the separating share falls in every cell where a window exists,
by between $0.08$ and $0.23$: German negotiated software from $0.360$ to
$0.190$, British software from $0.516$ to $0.323$, German negotiated radiology
from $0.479$ to $0.253$, British audit from $0.749$ to $0.665$.  The lower
bounds of Table~\ref{tab:band} rise correspondingly, by between $55\%$
(German software, EUR 5{,}797 to 8{,}966) and $247\%$ (British radiology, EUR
3{,}218 to 11{,}151).  No qualitative statement in Sections~\ref{sec:wedge}
to~\ref{sec:dynamics} reverses: the ordering of configurations, the emptiness
of the German standard-terms and American cells, and the two-sided shape of the
band all survive.  What a positive posting cost does is move the small-ticket
boundary outward, which strengthens rather than weakens the exclusion result.
We report it as a sensitivity rather than adopting it, because the payoff the
equilibrium is derived from carries no such cost.

Alongside the numerical harness, the sign and monotonicity claims the
propositions rest on are re-derived symbolically from the printed formulas in
the replication package (\texttt{symbolic\_checks.py}, nine claims): the
cross-partial of \eqref{eq:cost} that yields single crossing, the shape
conditions on \eqref{eq:rho}, the convexity argument behind
Lemma~\ref{lem:uppertail}, the envelope condition and anchor of
\eqref{eq:lamF}, monotonicity of $Q$ on the calibrated power family, the sign
chain and total derivative behind Proposition~\ref{prop:payer}, and
$dh_{\min}/d\pi<0$ in Proposition~\ref{prop:paradox}.  The equilibrium
construction itself---existence, off-path beliefs and the deviation
arguments---is not of this kind and rests on the printed proofs.

\section{Implications}\label{sec:implications}

\subsection{For legal design: legislatures, courts and drafters}

The model prices one function of three familiar institutions and is silent on their others; the statements below are comparative statics within that function, not welfare verdicts.

\textbf{The cap floor.}  The result is a regime map, not a welfare verdict.  A
statutory floor below the litigation threshold is locally irrelevant.  Crossing
that threshold removes the zero-exposure pool but simultaneously raises the
anchor of positive messages.  Because these effects move opposite endpoints,
no unconditional monotonic welfare or width claim follows.

\textbf{Pooled indemnity.} Displacement and pooling enter through the same parameter, and both boundaries move with it: $(1-\psi)$ scales the enforcement floor $F$ and the ceiling $C$ alike, while the slope $\chi/\theta\pi$ of the schedule is untouched. The separating share is therefore not monotone in $\psi$. Where the ceiling is slack, raising $\psi$ lowers $F$ and with it the lower cutoff $a$, and the share rises: in British software development it climbs from $0.516$ at $\psi=0$ to $0.739$ at $\psi=0.492$, the cutoff $a$ falling from $0.442$ to $0.252$ while $c^{*}$ stays at $\alpha$. Only once the ceiling binds does the share fall, reaching zero at cell-specific closure thresholds $\psi^{\dagger}$ of $0.43$ (legal research), $0.55$ (radiology), $0.68$ (audit) and $0.82$ (software development) on the British cost vector. The three occupations whose ceiling already binds at $\psi=0$ are monotone throughout; software development is the one that is not. The calibrated NHS level, $\psi=0.90$, exceeds every threshold, which is why that channel admits no window for any occupation despite the absence of immunity. The retained-risk instruments of \citet[Prop.~8, retained and individually priced risk]{bauer2026a}---a deductible or an experience-rated premium---act within the model as reductions of $\psi$, re-opening the window once $\psi<\psi^{\dagger}$. The model thereby locates the design margin; it does not rank pooling against the risk-spreading benefits that motivate it, which are outside the model.

\textbf{Publicly financed human services.} At individual scale no monetary-exposure equilibrium exists for social work in any calibrated jurisdiction, and Table~\ref{tab:scale} identifies the margin that would have to move: within the model, the window opens only when the pledging balance sheet is roughly thirty times an individual professional's, and closes again under official liability even then. The instruments that remain type-dependent at individual scale are the non-monetary ones the model does not price---criminal guarantor duties, professional-regulatory sanction, directly procured engagement shares. We read this as directing the model's extension rather than as a policy conclusion.

\subsection{For firms}

\textbf{Locate both boundaries before investing.}  Verifiability can relax the
upper constraint while worsening separation from the zero-exposure pool.  A
firm must therefore identify both message topology and binding ceiling before
inferring a signaling return from audit-trail investment.

\textbf{Treat the engagement share as a capacity decision.} Equation~\eqref{eq:hmin} converts capability maintenance into a schedulable quantity. A radiology department targeting a stable unassisted-reading capability must reserve roughly nine cases in ten for human-first reading; a software organization roughly three in four. These reservations should enter capacity planning as a constraint, alongside utilization targets, rather than being left to the discretion of individual professionals whose short-run incentives run the other way.

\textbf{Do not certify before the record does.} In the agent-based market of \citet{bauer2026b} the building provider wins client share long before it posts a pledge, which implies that certification expenditure is wasted where a public performance record already separates. The corollary is that firms should invest first in making their record legible and verifiable, and only then in contractual exposure.

\subsection{Testable predictions}

The framework generates predictions that do not follow from technology-exposure accounts:

\begin{enumerate}[leftmargin=*,itemsep=1pt]
\item Liability multiples in AI-assisted professional service contracts should be ordered by configuration, not by legal family: highest in German negotiated merchant contracts, intermediate in English agreed-damages clauses, and bounded by the moderation regimes in Austrian and Swiss contracts, with the Austrian distribution exhibiting a discontinuity at the 2007 contract vintage (Proposition~\ref{prop:penalty}; Corollary~\ref{cor:vintage}).
\item Within the United Kingdom, liability-based quality signaling should be present in private healthcare and absent in the NHS at identical clinical activity---a quasi-experiment on pooling with the immunity channel held fixed; case mix, procurement and regulatory oversight also differ across these channels, so the comparison identifies the pooling channel only under controls for clinical activity.
\item Within a fixed message topology, changes in verifiability should move the
two observed cutoffs in opposite directions; this is a sharper test than an
unconditional country ranking (Proposition~\ref{prop:theta}).
\item Deskilling should appear first and fastest at both ends of the transaction-size distribution and slowest in the mid-market, rather than ordered by automatability \citep[extending][]{bauer2026b}.
\item Measured unassisted performance should deteriorate fastest in occupations with the \emph{lowest} AI failure rates (Proposition~\ref{prop:paradox}), the opposite of a naive exposure ranking.
\item In publicly financed human services, liability-based quality commitments should be observed only at organisational scale and only among contracted independent providers, never among public-authority staff performing identical work.

\item Across the Austrian reform, the pre-reform separating set should weakly
contain the post-reform set within the same contracting channel
(Corollary~\ref{cor:vintage}); no cross-occupation magnitude ordering is claimed.

\item Contracting-channel changes should alter the topology of observed
commitments when they move the statutory floor across the litigation threshold
(Lemma~\ref{lem:msg}).

\end{enumerate}

\section{Limitations and conclusion}\label{sec:conclusion}

Four limitations bound the claims. \emph{Calibration is not estimation.} We place parameters ordinally against published evidence; the regime classifications are robust to substantial perturbation, but the point values in Table~\ref{tab:heatmap} should not be read as measurements. \emph{The jurisdiction vectors are stylized types.} The United States is modelled as a point when it is a distribution over fifty state regimes; the channel and vintage indexation sharpens the DACH vectors but inherits their stylization---the standard-terms multiple compresses heterogeneous case-group ceilings into one number, and the doctrinal multiples $m^{\mathrm{doct}}$ abstract from the enforcement practice that determines the effective multiples $m^{\mathrm{eff}}$, a divergence largest in Switzerland, where reduction requires the debtor's substantiated plea, and the object of companion empirical work \citep{bauer2026g} rather than of this model. \emph{The rescue technology is assumed, not estimated.} As \citet{bauer2026b} notes, estimating $\rho(\cdot)$ from incident data---aviation, radiology, clinical decision support---remains the single most valuable empirical extension, and the concavity of $\rho$ drives the shape of the window. \emph{The AI frontier is exogenous.} We hold $\alpha$ and $\pi$ fixed, whereas human correction of AI failures feeds back into model improvement. The omission is conservative with respect to our own conclusion: if $\pi$ falls as engagement rises, then by Proposition~\ref{prop:paradox} the engagement requirement $h_{\min}$ rises with it, so a model with that feedback would make capability maintenance harder, not easier. \emph{Non-monetary commitments are outside the model.} Criminal guarantor liability and professional-regulatory sanction are type-dependent costs and therefore belong to the instrument class Proposition~\ref{prop:sc} identifies; pricing them would likely reopen part of the social-work cell.

Three findings run against the picture developed here, and we state them as
facts rather than as qualifications.

\emph{The observed distribution of posted caps is nearly degenerate, and it is
keyed to fees rather than to the value at stake.} In a benchmark of 16{,}140
signed business-software agreements executed between June 2025 and June 2026,
$99\%$ carry a multiplier cap and $96\%$ of those cap liability at one times
the fees paid or payable, with sub-unity caps falling from $7\%$ in 2024 to
$2.5\%$ \citep{commonpaper2026}. A message that $96\%$ of senders choose
identically carries no information about the sender in any equilibrium. Remark~\ref{rem:multiple}
places the finding. Posting a multiple rather than an amount leaves the
equilibrium intact; bounding the multiple does not. A convention at one times
fees lies below $F/p(a)$ in every calibrated cell, so no type can reach the
schedule and the market pools---which is what a degenerate distribution of caps
is. The model's own prediction is the threshold $\bar k^{\dagger}=F/p(a)$ and,
beneath it, that posted exposure varies with capability within a segment. We are
not aware of a study that measures whether business buyers price a higher cap;
the conjoint design of \citet[Sec.~9]{bauer2026b} is addressed to that gap.

\emph{In the audit cell the statutory constraint has the opposite sign to the
one we model.} Section~\ref{sec:wedge} treats mandatory law as a floor.  For
statutory audit it is a ceiling, and a binding one: \S~323(2) HGB caps the
auditor's liability for ordinary negligence at EUR 1.5 million for ordinary
limited companies, EUR 4 million for CRR credit institutions and insurers and
EUR 16 million for capital-market-oriented public-interest entities---gross
negligence carries higher caps or none---and \S~323(4) provides that the
liability may be neither excluded nor limited by contract; \S~275(2) UGB sets
a graduated Austrian equivalent of EUR 2, 4, 8 and 12 million for ordinary
negligence, by size class, with the same non-derogation clause in \S~275(4).  These caps do not
scale with the engagement.  Appendix~\ref{app:sources} records that we note
$\bar L_{\mathrm{law}}$ for Germany and do not impose it; the reason is that at
the calibrated numeraire the solvency ceiling binds first in every audit cell,
so imposing the statutory cap changes no reported classification.  That is a
statement about the calibrated range and not about the institution.  At
engagement sizes where the statutory cap bites, the audit cell is governed by a
non-scaling ceiling that no provider can overcome, however skilled or solvent,
and the two-sided window of Proposition~\ref{prop:bunch} becomes a one-sided
constraint.  A model that treated floor and ceiling as two signs of the same
mandatory restriction would cover both cases; ours covers one and declares the
other.

\emph{The jurisdiction is a choice, not an endowment.} We index cells by
jurisdiction as though the applicable law were given. In a survey of 1{,}220
German firms and 193 arbitration practitioners, around three quarters of the
practitioners reported the choice of foreign law---predominantly Swiss---for
purely domestic matters, and about a third of contracts were concluded without
any negotiation at all \citep{leuschner2014}. Both cut at the model. If the
mandatory floor can be escaped by choosing another law, then $f$ is a decision
variable and the window is a statement about the chosen regime rather than the
assigned one; and where a third of contracts are not negotiated, the posting
decision the model prices is not being made. The first point is also an
opportunity the model does not exploit: a provider who selects the weaker
binding is, within this framework, signalling less, which is a testable
implication of a channel we treat as exogenous.

One scope restriction belongs with these rather than with the four limitations
above. The displacement parameter $\psi$ is calibrated to pooled premia, where
the insured portion of the promise carries no information about the insured.
Under experience rating the premium is itself decreasing in the insurer's
assessment of capability, and \citet[Prop.~8, retained and individually priced risk]{bauer2026a} shows that the
separation condition is then restored through the premium channel. Our $\psi$
prices the first case only. Where liability insurers underwrite, differentiate
premia and impose loss-prevention conditions, the displacement we model
overstates the loss of information, and the pooled-indemnity results of
Section~\ref{sec:implications} should be read as applying to genuinely pooled
regimes---state indemnity schemes and flat-rated professional cover---rather
than to insurance as such.

What the paper establishes is deliberately narrower.  Liability can certify
preserved human capability only on a separating interval whose lower endpoint
is determined jointly by the zero-exposure option and litigation viability and
whose upper endpoint is determined jointly by solvency, doctrine and pool
beliefs.  The interval cannot be recovered by intersecting an unconstrained
schedule with legal bounds.  The calibration illustrates this topology; it
does not establish legal-family rankings or welfare effects.  Whether loss of
the contractual signal also erodes capability remains outside the equilibrium
model and is not inferred from it.

\newpage
\newpage
\section*{Use of generative AI}
\small
\noindent In preparing this manuscript the author used Anthropic's Claude, including Claude Code (accessed July to September 2026), for four purposes: supplementary literature search and source discovery; drafting and language editing of the text; implementing and running the numerical and simulation code underlying the model calculations and the figures; and assistance with the derivations. The research question, the model and its assumptions, the interpretation of the results and all final decisions about content are the author's. AI technology was not used to generate, alter or fabricate research data, and is not credited as an author. No AI-assisted output entered the manuscript unverified: every reference was checked against the publisher record, and every reported number is regenerated from the deposited replication code rather than transcribed (see Data availability). The author takes full responsibility for the content of the manuscript, including all claims, citations and code.
\normalsize

\appendix
\section{Proofs}

\subsection*{A.1 Proof of Proposition~\ref{prop:sc}}
From \eqref{eq:cost}, $C(s,\Lam)=\theta\pi(1-\rho(s))\Lam$. Then $\partial C/\partial\Lam=\theta\pi(1-\rho(s))>0$ since $\rho(s)<1$ for $s<\alpha$, and $\partial^2C/\partial s\partial\Lam=-\theta\pi\rho'(s)<0$ since $\rho'>0$ by \eqref{eq:rho}. Strict monotonicity of $\Lam$ in $L$ on an interval $I$ implies $\partial^2C/\partial s\partial L=(\partial^2C/\partial s\partial\Lam)(\partial\Lam/\partial L)<0$ on $I$, which is the Spence--Mirrlees condition; existence of a least-cost separating equilibrium on $I$ then follows by the standard argument \citep{riley1979}, as in \citet[Prop.~2, the separating schedule]{bauer2026b}, whose client stage, posterior formation and off-path beliefs satisfying the intuitive criterion are inherited here rather than re-derived. The kinks of the wedge lie outside the strictly increasing segments by construction; implementability across them is the content of Lemma~\ref{lem:impl}. For the converse, if $\Lam$ is constant on the feasible range then $C$ is independent of $L$ there, so all types face identical costs of every message. In any PBE, a message profile in which types are distinguished would require some type to strictly prefer deviating to another's message whenever beliefs differ; with identical costs and identical benefits from a given belief, no such separating profile is incentive compatible, and every equilibrium pools. \qed

\subsection*{A.2 Proof of Proposition~\ref{prop:penalty}}
Two conditions must be separated. \emph{Recovery.} \citet[Prop.~3, the $\varphi D\geq v$ condition]{bauer2026a} establishes that
the client's expected recovery on a verified failure is $\theta G$, where $G$ is the gross
payment. Deterrence-adequate \emph{recovery} at value $v$ requires $\theta G\geq v$; separation is
a different requirement and is governed by Proposition~\ref{prop:bunch}, not by the condition
derived here. The penalty doctrine renders unenforceable any agreed sum exceeding $m$ times
compensatory loss, so $G\leq m\kappa v$ by Lemma~\ref{lem:impl} \emph{on the non-cumulation
branch}, which is the default branch of that lemma and the positive law of every civil-law
jurisdiction we calibrate. A gross payment satisfying both exists iff $\theta m\kappa v\geq v$,
i.e.\ $m\geq1/(\theta\kappa)$, which is \eqref{eq:penaltycond}. On the cumulation branch the
carriers add and the corresponding threshold is $m\ge(1/\theta-1)/\kappa$, as stated in
Proposition~\ref{prop:penalty}; the branch must be declared before either threshold is used. The
displacement factor is absent because a payment made by an insurer, pool or state on the
provider's behalf still reaches the client: $(1-\psi)$ reduces what the provider bears, not what
the client receives.

\emph{Signal.} The exposure that carries information is $R=(1-\psi)G$, and its admissible range
is what Proposition~\ref{prop:bunch} converts into a separating set. Displacement therefore
enters the signalling condition in full while leaving \eqref{eq:penaltycond} untouched. If
\eqref{eq:penaltycond} fails, $\theta$ can be raised to $1/(m\kappa)$ at cost, or an instrument
outside the doctrine can be used---one whose cost is type-dependent but which is not a
liquidated-damages clause, which lies in the admissible class by
Proposition~\ref{prop:sc}.\qed

% Agent-based appendix withdrawn; see Section 7.2.

\section{Complete parameter listing}\label{app:params}
Tables~\ref{tab:occ} and~\ref{tab:jur} report the parameters that carry the substantive
argument. For independent replication the remaining values are listed here; all are constant
across cells unless shown otherwise.

\begin{table}[h]\centering\small
\caption{Occupation parameters not reported in Table~\ref{tab:occ}.}
\begin{tabular}{lccccccc}
\toprule
Occupation & $a$ & $\zeta$ & $s_F$ & $\delta$ & $k_L$ & $\beta$ & $\sigma_v$\\
\midrule
Software / web dev. & 0.50 & 0.80 & 0.03 & 0.50 & 0.00 & 0.90 & 0.75\\
Legal research & 0.50 & 0.80 & 0.03 & 0.50 & 0.00 & 0.90 & 0.90\\
Radiology / imaging & 0.45 & 0.75 & 0.03 & 0.55 & 0.00 & 0.90 & 1.00\\
Audit \& assurance & 0.50 & 0.80 & 0.03 & 0.45 & 0.00 & 0.90 & 0.85\\
Social work & 0.55 & 0.60 & 0.03 & 0.40 & 0.00 & 0.90 & 1.40\\
\bottomrule
\end{tabular}
\vspace{2pt}
\footnotesize The fixed cost of posting $k_L$ is normalised to zero: the signalling payoff
\eqref{eq:signalpayoff} carries none, and a positive value would have to enter the lower-cutoff
condition as well as the participation comparison, both being the same comparison made by the
cutoff type (Remark~\ref{rem:participation}). Section~\ref{sec:robust} reports the sensitivity to
$k_L=0.02$. Common to all occupations: $b=2.0$ (ticket density shape), $\kappa_w=0.08$ (wage
share), $\bar w=0.15$, $\omega_{\mathrm{cont}}=2.0$. The payer-alignment weight is $\omega=1$ in
all private cells and enters only through Proposition~\ref{prop:payer}. The prevailing probability $\rho_p$ enters every participation root of Lemma~\ref{lem:uppertail}; it is held at $0.50$ throughout and stress-tested as a dedicated family in Section~\ref{sec:robust}.
\end{table}

\begin{table}[h]\centering\small
\caption{Jurisdiction parameters not reported in Table~\ref{tab:jur}.}
\resizebox{\textwidth}{!}{%
\begin{tabular}{lccccccccc}
\toprule
Jurisdiction & $c_1$ & $a_c$ & small-claims & $q_{\rm sc}$ & contingency & $q_{\rm att}$ & $\rho_p$ & $\psi_{\text{private}}$ & $\psi_{\text{public}}$\\
\midrule
Germany & 0.35 & 0.66 & --- & --- & --- & --- & 0.50 & 0.00 & 0.95\\
Austria & 0.33 & 0.66 & --- & --- & --- & --- & 0.50 & 0.00 & 0.95\\
Switzerland & 0.30 & 0.70 & --- & --- & --- & --- & 0.50 & 0.00 & 1.00\\
United Kingdom & 0.50 & 0.80 & 2.00 & 0.25 & --- & --- & 0.50 & 0.00 & 0.15\\
UK (NHS channel) & 0.50 & 0.80 & 2.00 & 0.25 & --- & --- & 0.50 & 0.90 & 0.90\\
United States (B2B) & 0.40 & 0.70 & --- & --- & 0.33 & 0.60 & 0.50 & 0.00 & 0.55\\
United States (B2C) & 0.40 & 0.70 & --- & --- & 0.33 & 0.60 & 0.50 & 0.00 & 0.55\\
\bottomrule
\end{tabular}
}
\vspace{2pt}
\footnotesize $\bar L_{\text{law}}=\infty$ throughout; the statutory auditor cap of
\S~323(2) HGB is noted in Appendix~C but not imposed in the baseline. $\psi_{\text{private}}$
governs the private cells reported in the main text and is zero except in the NHS channel;
$\psi_{\text{public}}$ governs the public-provider cells and is used only where the text says
so. The probability of prevailing is $\rho_p=0.5$ in the enforcement-viability condition.
$q_{\rm sc}$ is the below-threshold own-cost block as a share of $c_0$; it is
an explicit scenario primitive, not an empirical estimate, and Section~\ref{sec:robust}
perturbs it at $\times0.5/\times1.5$.
$q_{\rm att}$ is the contingency lawyer's case-cost block as a share of the
client's own-cost block and is treated identically.
\end{table}

\section{Calibration sources}\label{app:sources}
\subsection*{C.1 Provenance of model components}

Component provenance is stated in the main text (Table~\ref{tab:provenance}, Section~\ref{sec:provenance}).

\subsection*{C.2 Occupation parameter sources}

\textbf{Software / web development.} $\pi=0.30$: \citet{pearce2022} report 39.3\% of top suggestions and 40.7\% of all suggestions vulnerable across 1{,}689 programs generated in 89 scenarios; \citet{fu2023} report 35.8\% of Copilot-generated snippets in live GitHub projects carrying CWE weaknesses; a subsequent replication reports 27.3\%. These are security-flaw rates, not general error rates, and we use them as an upper-bounded proxy for the rate at which human fallback is required. $\gamma=0.55$: \citet{brynjolfsson2025canaries}, ADP payroll data on 4.6 million workers, find a 16\% relative employment decline for ages 22--25 in exposed occupations (revised from 13\% in the August 2025 version) and roughly 20\% for software developers. $\theta_0=0.80$: git blame, test suites, reproducible builds, formal acceptance under \S~631 BGB. $\bar L_{\text{solv}}=3.0$: German IT professional-indemnity cover is written at EUR 250{,}000--10 million and is not sublimited to contract value.

\textbf{Legal research.} $\pi=0.25$: \citet{magesh2025}, preregistered, find 17\% (Lexis+ AI) and 33\% (Westlaw AI-Assisted Research) hallucination rates; \citet{dahl2024} find 58--88\% for general-purpose models. $\theta_0=0.50$: citations are verifiable but the causal link from advice to loss is contested.

\textbf{Radiology.} $\gamma=0.45$: \citet{budzyn2025} find adenoma detection in non-AI colonoscopy falling from 28.4\% (226/795) to 22.4\% (145/648), an absolute difference of $-6.0$pp (95\% CI $-10.5$ to $-1.6$), among endoscopists with $>2{,}000$ prior procedures. The causal interpretation is contested in peer-reviewed correspondence and the article carries a corrigendum restoring an omitted covariate (indication for colonoscopy); independent commentary notes a near-doubling of procedure volume over the study window as a confounder. We therefore treat the study as placing $\gamma$ ordinally high, not as identification. $\theta_0=0.70$: PACS archives and ground truth support attribution, but \citet{dratsch2023} show automation bias degrades the human check itself---correct ratings fall from 79.7\% to 19.8\% among inexperienced readers under incorrect AI suggestions.

\textbf{Audit.} Ordinal placement. $\theta_0=0.85$ reflects regulated working-paper documentation; $\bar L_{\text{law}}$ is finite in Germany under \S~323(2) HGB, which we note but do not impose in the baseline.

\textbf{Social work / child protection.} $\alpha=0.35$: the core service is relational; AI is confined to screening and documentation. $\pi=0.35$: \citet{moore2025} find counselling chatbots giving unsafe responses in roughly 20\% of crisis scenarios against 7\% for human therapists; documented disability and racial bias in child-protection screening tools supports a high effective failure rate in the decision-support role. $\theta_0=0.25$: formal documentation under \S~36 SGB VIII is extensive, but forensic causal attribution is weak and responsibility diffuses across agencies. $\bar L_{\text{solv}}=0.5$ against $\bar v=8$: harm is decoupled from the fee, which is set per \emph{Fachleistungsstunde} at approximately EUR 55 ($=0.011$ numeraire).

\subsection*{C.3 Jurisdiction parameter sources}

\textbf{Germany.} Loser pays, \S~91 ZPO; degressive GKG/RVG fee schedules give a fixed block of roughly EUR 2{,}000 and an exponent $a_c\approx0.66$. Contractual penalties are enforceable, \S\S~339~ff.\ BGB; between merchants an \emph{individually agreed} penalty escapes judicial reduction (\S~348 HGB, barring \S~343 BGB), while standard-terms penalties remain subject to \S~307 review with case-group ceilings (construction: 5\% of the contract sum, BGH VII ZR 210/01, BGHZ 153, 311; tightened VII ZR 42/22) and \S~138 BGB as the outer bound (in exceptional cases \S~242, BGH I ZR 168/05). We set $m=\infty$ as the stylized benchmark for the negotiated channel; consistently with the standard-terms scope condition of Remark~\ref{rem:scope}, supra-compensatory exposure in Germany is predicted to travel through negotiated clauses, not standard forms; the two channels are reported separately throughout (Table~\ref{tab:jur}); an across-channel hybrid would mix a floor and a ceiling that cannot attach to the same clause. Standard-terms control, \S\S~305 ff., 307 BGB: caps on liability for the breach of cardinal obligations are void where the ceiling fails to cover the contract-typical, foreseeable damage (BGH VIII ZR 155/99, BGHZ 145, 203, 216 = NJW 2001, 292; confirmed VIII ZR 337/11, NJW 2012, 3229, which requires no numerical sum---an abstract limitation to foreseeable contract-typical damages suffices, so the numeric floor $f v$ is a modelling device, not a legal quantum). Personal-injury and gross-fault caps in B2B standard terms are regularly invalid via \S~307(1), (2) no.~2 BGB, with \S~309 no.~7 carrying indicative effect (\emph{Indizwirkung}; BGH VIII ZR 141/06, NJW 2007, 3774, 3775). A void cap falls away entirely rather than being reduced---no \emph{geltungserhaltende Reduktion} (BGHZ 84, 109, 114~ff.\ = NJW 1982, 2309)---and the gap is filled by dispositive law, \S~306(2) BGB: this is the void-snap that truncates the enforceable choice set. Adequacy is measured against the contract-typical foreseeable damage and the drafter's control of the risk (X~ZR~133/03, NJW 2005, 422), not against the price of the service---supporting a floor that scales with the value at stake rather than with fees. Individually negotiated caps escape the control (\S~305(1) s.~3 BGB), but the case law's demanding reading of negotiation makes genuine escapes rare. We set $f=0.60$ as an illustrative benchmark consistent with ordinal legal severity; the legal sources identify the mechanism and the ranking, not a cardinal value. The floor binds in a cell iff $f>\max\{c_0,\underline{\Lam}_{\text{enf}}\}/((1-\psi)\bar v)$; with the German enforcement floor $\underline{\Lam}_{\text{enf}}=0.668$ the threshold is $0.668$ in the software cell and $0.27$, $0.22$ and $0.22$ for legal research, radiology and audit. The software benchmark $f=0.60$ therefore lies \emph{below} its threshold---in that cell the enforcement floor overtakes the mandatory floor, the mechanism behind the ranking reading of R2a in Section~\ref{sec:robust}---while the floor binds in the three larger-ticket occupations; the Austrian value $0.45$ sits below its own software threshold of $0.621$ in the same way. These threshold statements are made at the calibrated prevailing probability $\rho_p=0.50$ (Appendix~B); the participation roots scale in $\rho_p$, and the German software floor re-binds for $\rho_p\ge 0.520$ (Section~\ref{sec:robust}). Official liability, Art.~34 GG with \S~839 BGB, transfers liability to the state with discharging effect and confines recourse to intent and gross negligence, hence $\psi_{\text{public}}=0.95$ rather than $1$.

\textbf{Austria.} \emph{Amtshaftungsgesetz}: organ liability toward the injured party is excluded, with internal recourse only; $\psi_{\text{public}}=0.95$. Penalty doctrine: \S\S~348 and 351 HGB were repealed by the Handelsrechts-\"Anderungsgesetz (BGBl I 2005/120) with effect from 1~January 2007; since then \S~1336(2) ABGB applies to entrepreneurs, and \S~907(18) UGB (with Art.~XXXII(1) HaR\"AG for \S~1336 ABGB) attaches the regime to the date of contract formation, not of judgment. The prevailing view holds the restored moderation right mandatory in favour of entrepreneurs (Schauer, \"OJZ 2006, 73~f.; Danzl/Karner in KBB, \S~1336); a dissenting view treats it as waivable in advance (Reischauer, Zak 2021/7). The Supreme Court keys the applicable regime to the date of contract formation: for a 1990 contract it applied the \S~348 HGB bar and enforced the penalty in full (OGH 18 Nov.\ 2014, 4 Ob 184/14h), and for a 1994 contract it held the old law ``weiterhin zu beachten''---permitting moderation only because the promisor had contracted as a private person, not as a merchant (OGH 25 Nov.\ 2020, 6 Ob 219/20d). The latter decision remarks, obiter for its pre-reform facts, that since the HaR\"AG repeal the moderation right is ``generell unabdingbar''; the Court has not decided advance waivability in a case governed by the new law, and the dissenting view---published after that decision---does not treat the question as settled. Moderation is anchored in the actual or ex-ante possible damage rather than in a multiple of the contract value; observed outcomes range from no reduction to reduction by half. We therefore report $m^{\mathrm{post}}=1.5$ as a central value with a band of $1.2$--$2.0$ and full band sensitivity, and $m^{\mathrm{pre}}=\infty$ for pre-2007 contracts. Case-level verification of regime assignment: \citet{bauer2026g}. Standard-terms control: \S\S~879(3), 864a ABGB apply in B2B without reduction-to-validity; $f=0.45$ is an ordinal placement consistent with a regime stricter than the Swiss and laxer than the German, not a measured quantity, and attaches to the standard-terms channel only.

\textbf{Switzerland.} \emph{Verantwortlichkeitsgesetz} art.~3(3) excludes direct action against the official, hence $\psi_{\text{public}}=1$---the sharpest displacement in the sample. Penalty doctrine: the Code of Obligations has never known a merchant exception; art.~163(3) OR is mandatory (BGE 143 III 1 E.~4.1, referring to BGE 133 III 201 E.~5.2). Reduction proceeds only on the debtor's plea; the debtor bears the burden of pleading and proof (BGE 143 III 1; BGE 114 II 264 E.~1b; BGE 133 III 43 E.~4.1); courts intervene only at gross disproportion and with restraint. No calibrated point multiple exists; the column is a band result over $m_{\mathrm{CH}}\in[1.5,3.0]$ with cell-wise invariance thresholds (Corollary~\ref{cor:mx}), and the divergence between the doctrinal and the effective multiple is measured in \citet{bauer2026g}. Standard-terms control in B2B is confined to art.~100 OR (nullity of waivers for unlawful intent or gross negligence) and the unusualness rule; $f=0.20$ is an ordinal placement---the weakest floor of the three---not a measured quantity, standard-terms channel only.

\textbf{United Kingdom.} English rule, but costs are generally not recoverable on the Small Claims Track (approximately GBP 10{,}000)---save for fixed costs, court fees, witness expenses and a summary award against a party who has behaved unreasonably, CPR 27.14(2)---producing a step rather than a smooth threshold. The below-threshold block is parameterised as $q_{\rm sc}c_0$ with central $q_{\rm sc}=0.25$. This coefficient is a scenario input, not a legal or empirical estimate; it is reported in Appendix~B and perturbed at $\times0.5/\times1.5$. Low caps in B2B contracts are frequently upheld under the UCTA 1977 reasonableness test between parties of comparable strength (\emph{Watford Electronics v Sanderson} [2001] EWCA Civ 317; \emph{Goodlife v Hall Fire} [2018] EWCA Civ 1371, insurability as a central factor), though the test is fact-sensitive and caps do fail (\emph{St Albans v ICL} [1996] 4 All ER 481); no lower bound on caps arises from law. Penalty doctrine: after \emph{Cavendish Square Holding BV v Talal El Makdessi} [2015] UKSC 67 at [32], the test is whether a secondary obligation imposes a detriment out of all proportion to any legitimate interest in enforcement of the primary obligation. The verified post-2015 span shows the tolerance for supra-compensatory exposure to be substantial: in \emph{Houssein v London Credit Ltd} the Court of Appeal corrected the test ([2024] EWCA Civ 721) and, on remission, a default rate at four times the contractual rate was held non-penal ([2025] EWHC 2749 (Ch)), affirmed on a second appeal ([2026] EWCA Civ 830). That figure is a ratio of interest rates, not a ratio of agreed damages to compensatory loss, and it therefore does not identify a general exposure multiple for professional services; we use it only to establish that the post-\emph{Cavendish} tolerance for supra-compensatory exposure is substantial, and treat $m$ itself as a scenario parameter reported with its band. We set the central multiple $m=3$ with a defensible range of $1.2$--$4$; no point value within the range is defensible as a measurement, so British results are reported at the centre with band sensitivity in Section~\ref{sec:robust}. Crown Proceedings Act 1947 leaves the Crown broadly liable as a private person, so $\psi_{\text{public}}=0.15$. The NHS channel is modelled separately with $\psi=0.90$ to represent state-pooled indemnity without provider-specific premium differentiation.

\textbf{United States (B2B).} American rule; each side bears its own costs, with a contingency route available at a 33\% share. The attorney's case-cost block is $q_{\rm att}=0.60$ of the client's own-cost block. Like $q_{\rm sc}$, this is a disclosed scenario input rather than an estimate and is perturbed at $\times0.5/\times1.5$. UCC Article 2 governs goods only (\S~2-102; hybrid transactions via the predominant-purpose test, \emph{Bonebrake v.\ Cox}, 499 F.2d 951, 960 (8th Cir.\ 1974)); the expert services modelled here fall under state common law, where caps between sophisticated commercial parties are broadly enforced subject to carve-outs for gross negligence and willful misconduct (the New York line: \emph{Kalisch-Jarcho}, 58 N.Y.2d 377, 385; \emph{Sommer}, 79 N.Y.2d 540, 554; \emph{Abacus}, 18 N.Y.3d 675, 683). Restatement (Second) of Contracts \S~356 and UCC \S~2-718 bar supra-compensatory agreed damages, hence $m=1$. The US vector is a stylized common-law benchmark, not a fifty-state survey. Near-universal arbitration clauses in B2B suppress class aggregation, $\kappa=1$; the B2C configuration sets $\kappa=8$. $\psi_{\text{public}}=0.55$ parameterizes only the public-provider channel (federal- or state-employed professionals): the Federal Tort Claims Act (28 U.S.C.\ \S\S~1346(b), 2671--2680) waives federal sovereign immunity subject to the discretionary-function exception (\S~2680(a)), with heterogeneous state and municipal immunities alongside. The private B2B cells reported in the text use $\psi=0$ throughout. Freedom of contract in B2B implies no mandatory liability floor.

\paragraph{Observed cap distributions.} Industry contract benchmarks (Common Paper's cloud-service-agreement corpus; not peer reviewed) report roughly 85\% of agreements capping liability at one year of fees, about 13\% carrying enhanced caps typically up to five times contract value, and about 1\% uncapped. We use this distribution to discipline $m$ and note that a pronounced mass point at a conventional low cap with a thin uncapped tail is qualitatively consistent with a compressed message space; we do not treat it as a test.

\subsection*{C.4 What we do not claim}

We do not claim that the jurisdiction vectors are estimates. They are stylized types, and the United States in particular is a distribution over fifty regimes represented here as a point. We do not claim that the EU Product Liability Directive (EU) 2024/2853 alters the analysis for the providers we model: it covers software and AI systems as \emph{products} and does not reach negligently rendered professional services, does not compensate pure economic loss, and confines recovery to natural persons---so the B2B expert-service relationships studied here remain governed by national contract and tort law. The withdrawn AI Liability Directive is not treated as live law.

\clearpage
\setcounter{page}{1}
\renewcommand{\thepage}{A.\arabic{page}}
\begin{center}{\Large\bfseries Appendices}\end{center}

\end{document}